\documentclass[10pt, conference]{IEEEtran}

\usepackage{amsthm,amsmath,amsfonts,braket,algorithm,graphicx, braket,balance}

\theoremstyle{definition} \newtheorem{define} {Definition} [section]
\newtheorem {theorem} {Theorem}

\newtheorem {lemma} {Lemma}

\usepackage{ifthen}
\newboolean{fullversionflag}
\setboolean{fullversionflag}{true}
\newcommand{\fullversion}[2]{\ifthenelse{\boolean{fullversionflag}}{{#1}}{{#2}}}

\newcommand{\num}{\#}

\newcommand{\kb}[1]{\left[#1\right]}

\newcommand{\al}{\mathcal{A}}

\newcommand{\trd}[1]{\left|\left| #1 \right| \right|}

\newcommand{\st}{\text{ } : \text{ }}
\newcommand{\stt}{\text{ s.t. }}

\newcommand{\Hmin}{H_\infty}

\newcommand{\leakEC}{\texttt{leak}_{EC}}

\newcommand{\up}[1]{^{({#1})}}

\floatname{algorithm}{Protocol}

\newcommand{\uni}[1]{\mathcal{U}_{{#1}}}

\newcommand{\samp}{\Psi}

\newcommand{\Bell}{\mathcal{B}}
\newcommand{\bit}{\texttt{bt}}
\newcommand{\phase}{\texttt{ph}}

\setboolean{fullversionflag}{false}

\begin{document}

\title{New Key Rate Bound for High-Dimensional BB84 with Multiple Basis Measurements}

\author{
    % Author name
  \IEEEauthorblockN{Trevor N. Thomas}
  \IEEEauthorblockA{
    \textit{University of Connecticut} \\
        School of Computing\\
        Storrs, CT USA \\
        trevor.thomas@uconn.edu
    }    
    \and
    \IEEEauthorblockN{Walter O. Krawec}
    \IEEEauthorblockA{
      \textit{University of Connecticut} \\
      School of Computing\\
        Storrs, CT USA \\
        walter.krawec@uconn.edu
    }    
}

\maketitle

\begin{abstract}
  In this paper we derive a new bound on the secret key-rate of the High Dimensional BB84 protocol operating with multiple mutually unbiased bases (MUBs).  To our knowledge, our proof is the first for this protocol that is both general (in that it can handle arbitrary, asymmetric channels), and also the first that derives a bound on the quantum min entropy for general attacks, without relying on post selection techniques or the asymptotic equipartition property.  Because of this, our new result shows that far more optimistic key-rates are possible for a low number of signals, even in general channels.  Furthermore, our proof methods may be broadly applicable to other protocols relying on multiple measurement bases and we prove several technical lemmas that may have independent interest.  We evaluate our new bound and compare to prior work, showing that higher key-rates are possible in several operating scenarios.  We also show some interesting behavior of the protocol when faced with asymmetric noise.
\end{abstract}

\section{Introduction}%\label{sec:introduction:}

Quantum key distribution (QKD) is, by now, a rich and fascinating field of study.  At its core, it allows for the establishment of a secure secret key between two parties, the security of which does not depend on any computational assumptions (which are always needed for classical key distribution and classical public key cryptography to function).  See \cite{QKD-survey,pirandola2020advances,amer2021introduction} for a general survey.  However, beyond this, its study has also led to breakthroughs in other areas of cryptography, including the development of novel cryptographic applications.  For instance, certified deletion \cite{broadbent2020quantum} is a cryptographic primitive not possible using classical data alone, but its study and security proof, follows QKD methodologies.  Other, similar, instances also exist \cite{broadbent2016quantum,bozzio2024quantum}.  This demonstrates the importance and relevance of studying QKD as the methods and proof techniques often see broader application beyond key distribution.

In this paper, we revisit a high-dimensional (HD) QKD protocol introduced originally in \cite{cerf2002security}.  HD-QKD protocols have been shown to exhibit many advantages over their two-dimensional counterparts, especially in terms of theoretical efficiency and noise tolerance \cite{bechmann2000quantum,cerf2002security,sheridan2010security,vlachou2018quantum,iqbal2021analysis}.  See also \cite{cozzolino2019high} for a survey of HD-QKD.  In the protocol we consider, Alice and Bob communicate not using two-dimensional qubits, but, instead, $d$-dimensional qudits.  Furthermore, they send, and measure, these qudits using multiple mutually unbiased bases (MUBs).  We focus on the case where $d$ is prime, and we utilize $d+1$ MUBs which generalizes the six-state BB84 protocol (which uses $d=2$ dimensional systems and $d+1=3$ bases) to higher dimensions.  In particular, when $d=2$ (i.e., the six-state BB84), we know that the noise tolerance of this protocol is $12.6\%$, which is higher than the standard $11\%$ noise tolerance of the two-basis BB84 protocol \cite{renner2005information}.  This trend (increasing noise tolerance with increasing dimension and MUBs) shows up in the HD version of this protocol in the asymptotic case \cite{sheridan2010security}, and also in the finite-key regime \cite{wyderka2025high}.

In this paper, we derive a novel finite key proof of security for this HD-QKD protocol using Bouman and Fehr's sampling framework \cite{bouman2010sampling} as a foundation, along with proof techniques we developed for sampling based entropic uncertainty relations \cite{krawec2019quantum,yao2022quantum,krawec2023entropic}.  Our new proof of security is able to bound the quantum min entropy of the system directly, for any general coherent attack, without having to rely on alternative methods such as the asymptotic equipartition property (AEP) \cite{tomamichel2009fully}, or post selection techniques \cite{christandl2009postselection}, both of which tend to lead to significantly lower key-rates when the number of signals (i.e., protocol rounds) is small, and also when the dimension of the system is large.  Our method does not suffer these draw backs.  To our knowledge we are the first to prove security of this protocol in the finite key setting by directly bounding the quantum min entropy, for any arbitrary channel or attack.  Prior work in the finite key setting has either relied on AEP and von Neumann entropy bounds \cite{wyderka2025high}, or can only apply to symmetric channels \cite{wang2021tight}.  Our proof of security results in a key-rate that can work for any arbitrary (i.e., not necessarily symmetric or depolarizing), quantum channel or attack.

We make several contributions in this work.  First, we derive a new proof of security for the HD-QKD protocol using $d+1$ MUBs, with $d$ a prime.  Our methods here may be broadly applicable to other quantum cryptographic protocols (HD or otherwise) which require measurements in more than two bases.  Our proof method, combined with methods from \cite{yao2022quantum}, may also lead to new entropic uncertain relations involving multiple basis measurements, an area of open study still \cite{wyderka2025high}.  To derive our proof, we also discover and prove several technical lemmas which may be of independent interest.  For example, our Lemma \ref{lemma:result:security-chain} can be used in combination with Bouman and Fehr's main result, to easily utilize their analysis framework for any arbitrary QKD protocol, even those with abort conditions.

Finally, we evaluate our new security proof, and resulting key-rate, in a variety of scenarios, including scenarios involving asymmetric noise.  We show that for small signal sizes (i.e., small number of rounds) and/or high dimensions, our work can outperform recent work in \cite{wyderka2025high}.  For higher number of rounds, and/or lower dimensions, prior work in \cite{wyderka2025high} outperforms our result (though the two results converge in the asymptotic limit).  However, since both results are lower-bounds on the actual key-rate, our work can readily complement prior work; users simply need to take the maximum of our result and theirs to get a valid key-rate for any scenario.  While performing these evaluations, we make several interesting discoveries on the behavior of this HD-QKD protocol under asymmetric noise which mirrors behavior discovered in \cite{murta2020key} for the two-dimensional case.

\subsection{Preliminaries}\label{sec:main:prelim}

Let $\al_d = \{0,1,\cdots, d-1\}$ and let $q \in \al_d^N$.  Given any subset $t \subset \{1,2,\cdots, N\}$, we write $q_t$ to mean the substring of $q$ indexed by $t$ and we write $q_{-t}$ to mean the substring of $q$ indexed by the complement of $t$.   If $i$ is a number between $1$ and $N$, we write $q_i$ to mean the $i$'th character of $q$ and $q_{-i}$ to mean all other characters of $q$ except $i$.  Let $\num_j(q)$ be the number of times character $j$ appears in $q$ (for $j \in \al_d$).  We denote by $w(q)$ to be the \emph{relative} Hamming weight of $q$, namely $w(q) = \frac{1}{N}|\{i \st q_i \ne 0\}| = \frac{1}{N}\sum_{j>0}\num_j(q)$.  Given a subset $t$ of size $m$, with $t=\{t_1,\cdots, t_m\}$ and $t_1<t_2<\cdots < t_m$, and a string $s\in\al_d^m$, we write $t_{s=j}$ to mean those entries in $t$ where $s_i=j$.  Namely, $t_{s=j} = \{t_i \st s_i = j\}$.

Given some $d$-dimensional orthonormal basis $\mathcal{B} = \{\ket{b_0}, \cdots, \ket{b_{d-1}}\}$, and some character $i \in \al_d$, we write $\ket{i}^{\mathcal{B}}$ (with the basis label in the superscript of a ket) to mean $\ket{b_i}$ (i.e., to mean ``$i$'' in the $\mathcal{B}$ basis).  Given a word $q \in \al_d^N$, we write $\ket{q}^{\mathcal{B}}$ to mean $\ket{q_1}^{\mathcal{B}}\otimes\cdots\otimes\ket{q_N}^{\mathcal{B}}$.  If the basis is not specified, we mean the standard computational basis.  Two bases are said to be MUB if for any pair of vectors, one from each basis, the magnitude of their inner product squared is $1/d$.

We will be working with the high dimensional Bell basis, spanned by vectors denoted:
\begin{equation}%\label{eq:introduction:}
\ket{\phi_x^y} = \frac{1}{\sqrt{d}}\sum_{a=0}^{d-1}\omega^{ay}\ket{a,a+x},
\end{equation}
where the addition is performed modulo $d$ and were $\omega = \exp(2\pi i/d)$.  Let $\Bell = \{{y\choose x} \st x,y\in\al_d\}$.  Then, given $k = {y\choose x}\in\Bell$, we write $k^\bit$ to mean the bottom coordinate (i.e., $k^\bit = x$) while we write $k^\phase$ to mean the top coordinate, $k^\phase = y$ (here $\bit$ stands for the ``bit'' portion of the Bell state while $\phase$ stands for the ``phase'' portion).  Given this, we write $\ket{\phi_k}$ to mean $\ket{\phi_{k^\bit}^{k^\phase}} = \ket{\phi_x^y}$.  Similarly, given $k\in\Bell^n$, we write $\ket{\phi_k}$ to mean $\ket{\phi_k} = \ket{\phi_{k_1}}\cdots\ket{\phi_{k_n}} = \ket{\phi_{k_1^\bit}^{k_1^\phase}}\cdots\ket{\phi_{k_n^\bit}^{k_n^\phase}}$.  All indexing rules discussed above apply in the natural way to a word $k\in\Bell^n$.  Finally, given $k\in\Bell^n$, we write $\num_\alpha^\beta(k) = \num_\alpha^\beta{k^\phase \choose k^\bit}$ to be the number of $j$ such that $k_j^\bit = \alpha$ (the bottom coordinate) and $k_j^\phase = \beta$ (the top coordinate).

Let $\rho$ be a quantum state, i.e., a \emph{density operator}, acting on some Hilbert space $\mathcal{H}_A\otimes\mathcal{H}_E$; in that case, we usually write $\rho_{AE}$ with the block subscripts indicating the underlying spaces the operator acts on.  We write $\rho_A$ to mean the result of tracing out the $E$ system.  The same notation follows for three or more systems.  To compress notation, given a pure state $\ket{\psi}$, we write $\kb{\psi}$ to mean $\ket{\psi}\bra{\psi}$.  Given two density operators $\rho$ and $\sigma$, acting on the same space, we write $\trd{\rho-\sigma}$ to be the trace distance of the two operators.

Given (classical) random variables $A$ and $E$, we write $H(A)$ to mean the Shannon entropy of $A$ and $H(A|E)$ to mean the conditional Shannon entropy of $A$ conditioned on $E$.  We write $h_d(x)$ to mean the $d$-ary entropy function, namely $h_d(x) = x\log_d(d-1) - x\log_dx - (1-x)\log_d(1-x)$.  Given a quantum state $\rho_{AE}$ (where, now, the $A$ and $E$ systems may be quantum), we write $\Hmin(A|E)_\rho$ to be the conditional quantum min entropy defined in \cite{renner2008security}. Note the subscript indicates the operator we are evaluating the entropy on ($\rho$ in this case).  If the $A$ system is classical, then it was shown in \cite{konig2009operational} that:
\begin{equation}%\label{eq:introduction:}
\Hmin(A|E)_\rho = -\log_2\max_{\mathcal{E}_a}\sum_aPr(A=a)tr(\mathcal{E}_a\rho_E^{A=a}).
\end{equation}
where the maximum is over all POVMs over Eve's ancilla.
The \emph{smooth quantum min entropy} \cite{renner2008security}, denoted $\Hmin^\epsilon(A|E)_\rho$ is defined to be $\Hmin^\epsilon(A|E)_\rho = \sup_\tau\Hmin(A|E)_\tau$, where the supremum is over all operators $\tau$, acting on the same space as $\rho$, such that $\trd{\tau-\rho}\le \epsilon$.  Here $\trd{\cdot}$ is the trace distance.

As shown in \cite{renner2008security}, quantum min entropy is a vital resource to measure in quantum cryptography as it relates to how much secure uniform randomness may be extracted from the $A$ system of a classical-quantum (cq)-state $\rho_{AE}$.  In particular, if the $A$ register consists of $n$ (classical) bits and if Alice chooses a random two universal hash function $f:\{0,1\}^n\rightarrow \{0,1\}^\ell$ and applies it to her register and broadcasts the actual hash function used (creating random variable $F$), it holds that \cite{renner2008security}:
\begin{equation}\label{eq:main:PA}
\trd{\rho_{f(A),EF} - \uni{\ell}\otimes\rho_{EF}} \le 2^{-\frac{1}{2}(\Hmin(A|E) - \ell)},
\end{equation}
where $\uni{\ell} = I/2^\ell$, namely a uniform random $\ell$-bit string.

Given a QKD protocol, we can define its security as follows:
\begin{define}\label{def:secure}
(From \cite{renner2008security}): Given $\epsilon > 0$, and let $\rho_{ABE}$ be an input state for a QKD protocol.  The QKD protocol will either output two $\ell$-bit strings $K_A$ and $K_B$ for Alice and Bob respectively, or abort (which will be modeled mathematically as a zero operator).  Let $p_{ok}\cdot\rho_{K_AK_BE} + (1-p_{ok})\cdot \mathbf{0}$ be the resulting state where $p_{ok}$ is the probability the protocol did not abort.  Then the QKD protocol is said to be \emph{$\epsilon$-secure on $\rho_{ABE}$} if:
\begin{equation}%\label{eq:introduction:}
p_{ok}\trd{\rho_{K_AE} - \uni{\ell}\otimes\rho_{E}} \le \epsilon,
\end{equation}
A QKD protocol is said to be \emph{$\epsilon$-secure} if the above is true for any input state $\rho_{ABE}$. See \cite{renner2008security} for more details on the definition and reasoning behind it.
\end{define}
Note that the above does not directly address correctness of the protocol (namely that Alice and Bob will share the same key or abort) however this is achieved through error correction (which must be taken into account when determining the information Eve has before privacy amplification) and can be verified by hashing the raw key through a small two-universal hash to test correctness (leaking additional information).  More details can be found in \cite{renner2008security,tomamichel2012tight}.  In this work we are primarily concerned with proving the security QKD protocols, relying on standard methods to ensure correctness.

Quantum min entropy has several properties we will make use of later.  First, given a state $\rho_{AEZ}$ that is classical in $Z$ (namely, $\rho_{AEZ} = \sum_zp(z)\kb{z}_Z\otimes\rho_{AE}^{(z)}$), then it holds that:
\begin{equation}\label{eq:main:min-ent-classical}
\Hmin(A|E)_\rho \ge \Hmin(A|EZ)_\rho \ge \min_z\Hmin(A|E,Z=z)_\rho,
\end{equation}
where, by $\Hmin(A|E,Z=z)$ we mean $\Hmin(A|E)_{\rho^{(z)}}$, namely the conditional entropy of the $A$ and $E$ system, conditioned on outcome $Z=z$ occurring.

The following lemma will be important later as it will allow us to bound the min entropy of a pure state based on a particular mixed state:
\begin{lemma}\label{lemma:main:entropy-super}
  (From \cite{bouman2010sampling} based on a proof in \cite{renner2008security}): Let $\ket{\psi}_{AE} = \sum_{j\in J}\alpha_j\ket{j}^X\otimes\ket{E_j}$, for $J \subset \al_d^n$, and consider the mixed state $\chi = \sum_{j\in J}|\alpha_j|^2\kb{j}^X\otimes\kb{E_j}$.  Let $\rho_{ZE}$ be the result of measuring the $A$ system of $\ket{\psi}$ in some basis $Z$ (which need not be the computational basis) and $\sigma_{ZE}$ be the same but when measuring $\chi$.  Then it holds that $\Hmin(Z|E)_\rho \ge \Hmin(Z|E)_\sigma - \log_2|J|$.
\end{lemma}

\fullversion{The following is not difficult to show and is based on a proof in \cite{krawec2019quantum}:}
{The following is not difficult to show:}
\begin{lemma}\label{lemma:introduction:entropy-mixed}
  Let $\chi_{AE} = \sum_{i\in J}p_i\kb{i}_A^X\otimes\kb{E_i}$ for some values $p_i \ge 0$ and $\sum_ip_i=1$ where the $d$-dimensional $X=\{\ket{x_i}\}$ basis need not be the Hadamard basis and $J \subset \al_d^n$.  Assume a measurement of the $A$ system is made in some alternative $d$-dimensional basis $Z=\{\ket{z_i}\}$ (not necessarily the computational basis) resulting in random variable $Z$ and state $\sigma_{ZE}$.  Then it holds:
%  \begin{equation}%\label{eq:introduction:}
$    \Hmin(Z|E)_\sigma \ge n\cdot c,$
%  \end{equation}
  where $c = -\max_{i,j\in\al_d}\log_2|\braket{z_i|x_j}|^2$.
\end{lemma}
\fullversion{
  \begin{proof}%\label{pf:introduction:}
  This proof is taken from arguments in \cite{krawec2019quantum}.  Let $\chi_{AEI} = \sum_{i\in J}p_i\kb{i}^X\otimes\kb{E_i}\otimes\kb{i}_I$ be a state classical in the $I$ system.  Then, measuring the $A$ system in basis $Z$ yields:
  \begin{equation}%\label{eq:introduction:}
    \chi_{ZEI} = \sum_{i\in J}p_i \sum_{j\in\al_d^n}p(j|i) \kb{j}^Z\otimes\kb{E_i} \otimes \kb{i},
  \end{equation}
  where $p(j|i) = |\braket{x_i|z_j}|^2 = \prod_{\ell=1}^n|\braket{x_{i_\ell}|z_{j_\ell}}|^2$, namely it is the probability of observing $\ket{j}^Z$ given input state $\ket{i}^X$.

  From Equation \ref{eq:main:min-ent-classical} it holds that $\Hmin(Z|E)_\chi \ge \Hmin(Z|EI)_\chi \ge \min_{i\in J}\Hmin(Z|E, I=i)_\chi$.  Looking at the conditional state when $I=i$, we have:
  \begin{equation}%\label{eq:introduction:}
    \chi_{ZE,I=i} = \sum_jp(j|i)\kb{j}^Z\otimes\kb{E_i}\otimes\kb{i}.
  \end{equation}
  Since the $E,I=i$ system is now separable from the $Z$ portion, $\Hmin(Z|E,I=i)_\chi = \Hmin(Z)_{\rho\up{i}}$, where $\rho\up{i} = \sum_jp(j|i)\kb{j}^Z$.  But by basic definitions of min entropy \cite{renner2008security}, we have $\Hmin(Z)_{\rho\up{i}} = -\log\max_j p(j|i)$.  Concluding:
  %\begin{equation}%\label{eq:introduction:}
    $\Hmin(Z|E)_\chi \ge \min_i\Hmin(Z|E,I=i)_\chi \ge -\log\max_{i,j}p(j|i) \ge -n\cdot c.$
%  \end{equation}

  \end{proof}
  }{}

\subsection{Quantum Sampling}\label{sec:intro:sampling}

Our proof will use as a foundation a quantum sampling framework of Bouman and Fehr introduced originally in \cite{bouman2010sampling}.  This framework allows us to construct certain ``ideal'' states which behave nicely (to be formalized shortly) and which are $\epsilon$-close to the given input state, where $\epsilon$ will depend on the failure probability of a classical sampling strategy.  This framework has been used to analyze BB84 \cite{bouman2010sampling}, including the two-basis HD version of it in \cite{yao2022quantum}, and we have also had success in using it to prove novel entropic uncertainty relations \cite{krawec2019quantum,yao2022quantum,krawec2023entropic}.  Entropic uncertainty relations (see \cite{bialynicki2011entropic,coles2017entropic,wehner2010entropic} for a survey) are useful mathematical tools for proving security of quantum cryptographic protocols \cite{tomamichel2012tight} and our relations based on Bouman and Fehr's sampling methods have led to new relations which can sometimes outperform standard techniques \cite{yao2022quantum,krawec2023entropic}.  In this section, we briefly review some of the basic notions and results of this sampling method, referring the reader to \cite{bouman2010sampling} for additional details.

A \emph{sampling strategy} over words in $\al_d^N$ is a tuple $\samp = (P_T, g, r)$ where $P_T$ is a probability distribution over all possible subsets of $\{1, \cdots, N\}$ and where $g,r: \al_d^* \rightarrow \mathbb{R}^k$ (for some $k \ge 1$).  Here, ``$g$'' is a \emph{guess function} and ``$r$'' is a \emph{target function}. Given a word $q \in \al_d^N$, the goal of the strategy is to observe $g(q_t)$ (for some subset $t$ sampled according to $P_T$) and use that evaluation as a guess as to the target function evaluated on the \emph{unobserved} portion of the word, namely $r(q_{-t})$.  Note we are using a slightly extended definition of a sampling strategy used in \cite{yao2022quantum} (in the original \cite{bouman2010sampling}, $k=1$ always, however it is trivial to extend this to larger $k$ as discussed in \cite{yao2022quantum}).  Values $k > 1$ may arise when the sampling strategy looks at multiple statistics for example, as we require in our proof of security, later.

If $\samp$ is a good sampling strategy, it should hold that the guess and target functions are $\delta$ close to each other in each coordinate.  To formalize this, fix $\delta > 0$ and pick some subset $t$ such that $P_T(t) > 0$. We define the set of \emph{good words for subset $t$} for the given sampling strategy to be:
%\begin{equation}%\label{eq:introduction:}
$\mathcal{G}^t_\delta(\samp) = \left\{q\in\al_d^N \st \max_j|g_j(q_t) - r_j(q_{-t})|\le\delta\right\},$
%\end{equation}
where $g_j$ represents the $j$'th coordinate of the output of $g$ (similarly for $r_j$).  We will sometimes simply write $\mathcal{G}^t$ if $\delta$ and $\samp$ are clear from context. Note that other less stringent conditions may be considered (e.g., \emph{some} of the coordinates are good while others may be larger than $\delta$), however we do not consider that here.  In particular, given $q\in\mathcal{G}^t$, it is guaranteed that, should the sampling strategy actually choose subset $t$, it will hold that the guess and target functions will be $\delta$-close to one another in each coordinate.

\begin{define}\label{def:introduction:error-prob}
Given $\samp$ and $\delta > 0$, the \emph{error probability of $\samp$} is defined to be:
%\begin{equation}%\label{eq:introduction:}
$\epsilon_\delta^{cl} = \max_{q\in\al_d^N}Pr\left(q \not \in \mathcal{G}^t_\delta(\samp)\right),$
%\end{equation}
where the probability, above, is over the choice of subset $t$ according to $P_T$.  Note the ``$cl$'' superscript to denote that this is the error probability of a classical sampling strategy.
\end{define}
%In particular, given any $q\in\al_d^N$, the sampling strategy will produce a $\delta$-close guess of the target function, except with probability at most $\epsilon^{cl}_\delta$.

The above can all be extended in a natural way to quantum states.  Let $\mathcal{B}$ be some $d$-dimensional orthonormal basis and let $\ket{\psi}_{AE}$ be a quantum state where the $A$ system consists of $N$ qudits (each qudit of dimension $d$; thus $\dim A = d^N$).  The quantum sampling strategy will sample $t$ according to $P_T$ and measure those qudits indexed by $t$ in basis $\mathcal{B}$, resulting in a classical output $q_t\in\al_d^{|t|}$.  This allows the strategy to compute $g(q_t)$.  The main result of \cite{bouman2010sampling}, informally, is that the remaining unmeasured portion of $\ket{\psi}$ will behave as a superposition of ``good'' words, namely words that are $\delta$-close to $g(q_t)$.

Formally, we define the \emph{space of good words} to be:
%\begin{equation}%\label{eq:introduction:}
$\mathcal{G}^t_\delta(\samp;\mathcal{B}) = \text{span}\left\{\ket{q}^\mathcal{B} \st q \in \mathcal{G}^t_\delta(\samp)\right\}\otimes\mathcal{H}_E.$
%\end{equation}
Note the dependence on the chosen basis.  A state $\ket{\phi^t}_{AE} \in \mathcal{G}^t_\delta(\samp;\mathcal{B})$ is said to be an \emph{ideal state}.  %Note that if the sampling strategy were to choose $t$ and measure those qudits indexed by $t$ in basis $\mathcal{B}$ resulting in outcome $q\in\al_d^{|t|}$, it is guaranteed that the post measured state (tracing out the measured portion) will be of the form:
%$  \ket{\psi\up{t,q}} = \sum_i\ket{i}^{\mathcal{B}}\ket{E_i\up{t,q}},$
%where the sum is over all $i$ such that $\max_j|g_j(q_t) - r_j(i)| \le \delta$.

Given the above, the main result of Bouman and Fehr \cite{bouman2010sampling} was the following theorem:
\begin{theorem}\label{thm:introduction:sample}
  (From \cite{bouman2010sampling}, though modified slightly for our purpose here): Let $\delta > 0$, $\ket{\psi}_{AE}$ be a quantum state where $\dim A = d^N$, $\mathcal{B}$ be a $d$-dimensional basis, and $\samp$ be some classical sampling strategy with error probability $\epsilon^{cl}_\delta$.  Then there exist ideal quantum states $\ket{\phi^t}\in \mathcal{G}^t_\delta(\samp;\mathcal{B})$ such that:
  \begin{equation}%\label{eq:introduction:}
    \frac{1}{2}\trd{\sum_tP_T(t)\kb{t}\otimes(\kb{\psi}_{AE} - \kb{\phi^t}_{AE})} \le \sqrt{\epsilon_\delta^{cl}}.
  \end{equation}
  We typically call the joint state $\sigma_{TAE}=\sum_t\kb{t}_T\kb{\phi^t}_{AE}$ the \emph{ideal state.}
\end{theorem}

%% two party
Note that the above assumed a single party $A$ was performing the sampling.  However it is not difficult to see that the above may be generalized to a multi-party scenario $AB$ where $A$ and $B$ look at their individual portions of their systems and broadcast the outcome of their measurements (using, say, an authenticated channel).  That this all will follow, assuming one derives a suitable sampling strategy, from the above definitions is easy to see simply by considering the $AB$ system as a single party for the sake of sampling.  We later call such a setup a \emph{two-party sampling strategy}.  Finally, we note that the choice of subset $t$ may involve multiple random variables: for instance, a subset may be chosen, and then a subset within that subset using some additional seed randomness.  The above theorem still applies in this case.  See \cite{bouman2010sampling} for more details.

Before leaving, we discuss one basic sampling strategy, analyzed in \cite{bouman2010sampling}.  Let $\mathcal{G}^t_\delta = \{q\in\{0,1\}^N \st |w(q_t) - w(q_{-t})|\le\delta\}$ and let subsets $t$ be chosen uniformly at random from all subsets of size $m$, where $m \le N/2$.  Then it holds that the error probability of this strategy is $\max_{q\in\{0,1\}^N}Pr(q\not\in\mathcal{G}^t_\delta) \le \epsilon_0$, where:
\begin{equation}\label{eq:introduction:basic-sample}
\epsilon_0 = 2\exp\left(-\frac{\delta^2mN}{N+2}\right).
\end{equation}

%%% Local Variables:
%%% mode: LaTeX
%%% TeX-master: "main"
%%% End:

\section{Protocol}\label{sec:protocol:}
We now discuss the HD-QKD protocol we are analyzing, which was originally introduced in \cite{cerf2002security}.  To save space, we only define the entanglement-based version of the protocol; the prepare and measure version is easily derived from this version, and is also discussed in detail in \cite{cerf2002security}.

Let $d$, a prime number, be the dimension of the system used each round of the protocol.  In this case, it is known that the eigenvectors of the operators $Z$ and $XZ^k$, for $k=0, \cdots, d-1$ form a maximal set of MUBs \cite{wootters1989optimal,durt2010mutually}, where $Z = \sum_j\exp(2\pi i j/d)\kb{j}$ and $X = \sum_j\ket{j}\bra{j-1}$.  Abstractly, we index the bases used, by $j=0, 1, \cdots, d$ with basis $j=0$ being the computational basis (the eigenvectors of the $Z$ operator).  Let $\Lambda^j=\{\Lambda^j_0,\cdots, \Lambda^j_{d-1}\},$ be a POVM with elements $\Lambda^j_c = \sum_x\kb{x}^j_A\otimes\kb{x+c}^j_B$, where the arithmetic is done modulo $d$.  Namely, this models Alice and Bob measuring in the $j$'th basis, resulting in outcomes that are ``$c$'' apart (e.g., Alice observes $\ket{2}^j$ and Bob observes $\ket{2+c\text{ mod } d}^j$).  In practice, parties would measure their systems and report their outcomes over the authenticated channel, allowing them to compute $c$ for each testing round.  However, it is mathematically equivalent to consider the case where Alice and Bob measure in this POVM and report the difference in their outcomes ``$c$.''

The protocol consists of Eve (the adversary) preparing a quantum state $\ket{\psi}_{ABE}$, where the $A$ and $B$ systems each consist of $N$ qudits, each qudit of dimension $d$.  Alice and Bob choose a random sample $t$ of size $m$ from all $N$ possible systems to use as a Test.  They then choose a string $s \in \al_{d+1}^m$, with $Pr(s_i=j) = 1/(d+1)$ and measure their systems, indexed by $t_{s=j}$ in POVM $\Lambda^j$ (in practice they will measure their individual systems in the $j$'th basis, report the outcome, and compute $c$, the difference in outcomes; it is not difficult to show measuring in $\Lambda^j$ will produce an equivalent entropy bound for the key-rate analysis).  Measuring each $t_{s=j}$ using POVM $\Lambda^j$ produces outcomes $q^0,\cdots, q^d$, where $q^j \in \al_d^{m_j}$, with $m_j = |t_{s=j}| = \num_j(s)$.  Ideally, if there is no noise in the channel, it should hold that $w(q^j) = 0$ for all $j$.  If the noise is too high, parties will abort.

For those remaining $n = N-m$ systems, not measured in the Test stage, parties will measure their systems in the $Z$ basis (the $0$'th basis).  This results in their raw key.  Parties will run error correction and privacy amplification on their raw keys to produce their final secret key of size $\ell$-bits.  Our goal, in the next section, is to derive a bound on $\ell$, based on the observed noise $q^j$, the dimension of the system, and the number of rounds and the sample size.  For more details on the protocol, especially the mathematically equivalent prepare and measure version, the reader is referred to \cite{cerf2002security}.

%%% Local Variables:
%%% mode: LaTeX
%%% TeX-master: "main"
%%% End:

\section{Security Proof}%\label{sec:mainresult:}

We now prove our main result, namely a key-rate bound for the $(d+1)$-MUB protocol.  For this we will require a new sampling strategy and it's corresponding error probability.  We will also prove two technical lemmas, the first allowing one to easily use Theorem \ref{thm:introduction:sample} to analyze QKD protocols (or potentially other quantum cryptographic protocols relying on privacy amplification); the second allowing one to bound the min entropy of a particular state with respect to the HD-Bell basis. Both results may be of independent interest.  After this, we will be able to state and prove our main result, Theorem \ref{thm:mainresult:main-thm}.

\subsection{Promotion of the Ideal State Analysis}\label{sec:mainresult:ideal-promote}

We first discuss a small lemma which allows one to easily promote the analysis done of a QKD protocol given the ideal states produced by Theorem \ref{thm:introduction:sample}.  In particular, it allows one to readily bound the key size based on the quantum min entropy of the ideal states, conditioned on not aborting.  Our proof is general and does not rely on anything specific to the protocol we are considering in this paper, nor does it require any particular sampling strategy.  Since computing the min entropy of the ideal states is, in general, easier than the real states, this provides a ready framework for other QKD protocols.  Note that in prior work, to promote the ideal state analysis to the real QKD protocol, two methods were primarily used.  Either a probabilistic argument was made \cite{krawec2019quantum,krawec2022security} involving Chebyshev's inequality, however this led to sub-optimal results due to the need to raise the classical error probability to the $1/6$ power (approximately), thus requiring a large sample size to ensure a small $\epsilon$.  Alternatively, in Bouman and Fehr's original sampling work \cite{bouman2010sampling}, a direct method was derived to bound the keyrate, however they did not explicitly deal with an abort case.  Dealing with an abort is important in QKD research and our new lemma below can handle this scenario.

\fullversion{
We first need to state the scenario we are analyzing.  Let $\samp$ be a two-party sampling strategy and $\bar{\mathcal{S}}$ be the space of possible outputs of this strategy (i.e., possible outputs of the guess functions).  Let $\mathcal{S} \subset \bar{\mathcal{S}}$ be a user-specified set of ``acceptable'' outcomes of the sampling strategy; namely if the sampling strategy produces a guess of $s \not\in \mathcal{S}$, then the protocol will abort (e.g., if the noise is high for example).  Note there may be other conditions to cause parties to abort as discussed later.

A general QKD protocol (though this can be generalized to other quantum cryptographic protocols such as QRNG \cite{herrero2017quantum} for instance) will begin with some input state $\rho\up{0}_{TABE}$.  The $T$ register is a random variable induced by the subset choices of $\samp$.  Generally the $T$ and $ABE$ systems are separable, however when dealing with ideal states from Theorem \ref{thm:introduction:sample}, they will be entangled.  The sampling strategy is run to produce a mixed state $\rho\up{1}_{TSABE}$ where the $TS$ registers are the subset choice and the output of the sampling strategy.  Note the $AB$ systems are possibly smaller now since the sampling strategy measured some of the qubits of the original $AB$ registers.

Next, we apply a CPTP map $\mathcal{R}_\mathcal{S}$ called an Abort map which will set an ``abort'' flag to \texttt{true} (or ``$1$'') in a new (two dimensional classical) register $R$ if the value in the $S$ register is not in the given $\mathcal{S}$.  Otherwise, it sets the abort register to \texttt{false} (or ``$0$'').  The quantum state after this operation will be denoted $\rho\up{2}_{RTSABE}$.  It may also abort based on other user specified parameters (e.g., insufficient raw key bits available for instance).

After this the ``core'' QKD protocol is run.  This will run the given QKD protocol assuming the abort register is set to zero (\texttt{false}).  This map will also run an error correction protocol leaking at most $\leakEC$ bits of information (again, assuming the abort register is \texttt{false} or $0$).  If the abort register is \texttt{true} or $1$, then the protocol aborts which will be modeled mathematically as the zero operator as in \cite{renner2008security}.  The new state, which is now potentially sub-normalized, is $p_{ok}(\rho)\rho^*_{TSABE,R=0}$ where now the $A$ and $B$ registers contain the error corrected raw keys.  Note that $p_{ok}(\rho)$ is the probability that the protocol does not abort if the original input state is $\rho\up{0}$.  The operator $\rho^*_{TSABE, R=0}$ is normalized and the $R$ register is always $0$ (i.e., it is the state of the system conditioned on Alice and Bob not aborting).  The $E$ register now contains any error correction leakage.  Also note that we used a ``$*$'' superscript as this particular state will be important in our lemma as it is the state right before privacy amplification.

Finally, privacy amplification is run, which we denote as the CPTP map $\mathcal{P}_\ell$.  This will choose a random two-universal hash function with output size of $\ell$-bits, reveal it (i.e., Eve will know the choice also), and both Alice and Bob will apply the chosen hash function to their raw key bit registers resulting in their final secret key registers $K_A$ and $K_B$.  The choice of hash function will be stored in a classical register $F$. This creates the final state $p_{ok}(\rho)\rho_{TSK_AK_BEF,R=0}$.  Note the lack of superscript indicating this is the final state.}
{
  Let $\samp$ be a two-party sampling strategy and $\bar{\mathcal{S}}$ be the space of possible outputs of this strategy.  Let $\mathcal{S}\subset\bar{\mathcal{S}}$ be a user-specified set of acceptable outcomes (those for which parties will not abort).  A general QKD protocol can be modeled as a sequence of operations.  Starting with $\rho\up{0}_{TABE}$, where the $T$ register is induced by the sampling strategy and is generally independent of $\rho\up{0}_{ABE}$, the sampling strategy is run, yielding $\rho\up{1}_{TSABE}$, where $TS$ now contain the output of the strategy while the $AB$ registers may be smaller, now, as some qudits were measured.  Next, a CPTP map $\mathcal{R}_{\mathcal{S}}$, called an Abort map, is run which will set an ``abort'' flag to \texttt{true} or ``$1$'' in a new register $R$ if the sampling strategy produced an outcome not in $\mathcal{S}$.  Otherwise the abort flag is set to ``$0$.''  The new quantum state is denoted $\rho\up{2}_{RTSABE}$.

  Next, the ``core'' QKD protocol is run assuming the abort register is set to \texttt{false}, i.e., $R=0$.  This map will also run an error correction protocol, leaking $\leakEC$ bits of information.  If the abort register $R=1$, then the protocol simply aborts at this stage, which is modeled mathematically as the zero operator as in \cite{renner2008security} (note the protocol may also abort for other reasons at this stage).  The new state, which is now potentially sub-normalized, is $p_{ok}(\rho)\rho^*_{TSABE,R=0}$ where now the $A$ and $B$ registers contain the error corrected raw keys.  Note that $p_{ok}(\rho)$ is the probability that the protocol does not abort if the original input state is $\rho\up{0}$.  The operator $\rho^*_{TSABE, R=0}$ is normalized (it is the state of the system conditioned on Alice and Bob not aborting).  The $E$ system contains any error correction leakage.

  Finally, privacy amplification is run, denoted as the CPTP map $\mathcal{P}_\ell$.  This will choose a random two-universal hash function with output size of $\ell$-bits, store the choice in classical register $F$, and both Alice and Bob will apply the function to their raw key registers resulting in final secret keys $K_A$ and $K_B$.   This creates the final state $p_{ok}(\rho)\rho_{TSK_AK_BEF,R=0}$.  Note the lack of superscript indicating this is the final state.
  }

In summary, we have the following chain of maps:
\begin{equation}%\label{eq:mainresult:}
\rho\up{0} \xrightarrow{\Psi} \rho\up{1}\xrightarrow{\mathcal{R}_\mathcal{S}}\rho\up{2}\xrightarrow{\texttt{QKD}} p_{ok}\rho^* \xrightarrow{\mathcal{P}_\ell} p_{ok}\rho
\end{equation}
\fullversion{Our lemma, then, may be stated below:}{}
\begin{lemma}\label{lemma:result:security-chain}
  Let $\delta > 0$ and let $\rho\up{0}_{TABE}$ be an input state where the $T$ register is induced by the given sampling strategy $\samp$ (e.g., $\rho_T = \sum_tP_T(t)\kb{t}$), and is independent of the $ABE$ system.  Let $\sigma\up{0}_{TABE}$ be the ideal state from Theorem \ref{thm:introduction:sample} and let $\sigma^*_{TSABE,R=0}$ be the resulting operator conditioning on not aborting when running the QKD protocol as discussed.  Let:
  \begin{equation}%\label{eq:mainresult:}
    \Hmin(A|E,T=t,S=s,R=0)_{\sigma^*} \ge \gamma
  \end{equation}
  for some $\gamma$ and for every $t$ and $s\in\mathcal{S}$ (note that the $E$ system contains information on the error correction leakage which must be taken into account when determining $\gamma$).  Then, for $\ell \ge 0$, it holds that:
  \begin{equation}\label{eq:result:lemma-chain}
    p_{ok}(\rho)\trd{\rho_{K_AQ} - \uni{\ell}\otimes\rho_{Q}} \le 2^{-\frac{1}{2}(\gamma - \ell)} + 4\sqrt{\epsilon_\delta^{cl}}.
  \end{equation}
  where $Q$ is the joint system $Q = (TSEF,R=0)$
  That is, one may determine the security of the ``real'' QKD protocol if one can bound the min entropy of the ideal system, conditioned on not aborting, for every subset $t$ and every $s\in\mathcal{S}$.
\end{lemma}
\begin{proof}%\label{pf:result:}
  From Theorem \ref{thm:introduction:sample}, we know that $\frac{1}{2}\trd{\rho\up{0}_{TABE} - \sigma\up{0}_{TABE}} \le \sqrt{\epsilon^{cl}_\delta}.$
  Since CPTP maps do not increase trace distance, we have $\frac{1}{2}\trd{\rho\up{2}_{RTSABE}, \sigma\up{2}_{RTSABE}} \le \sqrt{\epsilon^{cl}_\delta}$.
  We may write these two states in the following form:
%  \begin{align}%
    $\rho\up{2}_{RTSABE} = p_{ok}(\rho)\kb{0}_R\otimes\rho\up{2}_{TSABE,R=0} + p_{rej}(\rho)\kb{1}_R\otimes\rho\up{2}_{TSABE,R=1}$ and
   $ \sigma\up{2}_{RTSABE} = p_{ok}(\sigma)\kb{0}_R\otimes\sigma\up{2}_{TSABE,R=0} + p_{rej}(\sigma)\kb{1}_R\otimes\sigma\up{2}_{TSABE,R=1}$
%  \end{align}

  From the above, and properties of trace distance, we have:
  \begin{align*}%
    \sqrt{\epsilon^{cl}_\delta} &\ge \frac{1}{2}\trd{\rho\up{2}_{RTSABE} - \sigma\up{2}_{RTSABE}}\\
                                &= \frac{1}{2}\trd{p_{ok}(\rho)\rho\up{2}_{TSABE,R=0} - p_{ok}(\sigma)\sigma\up{2}_{TSABE,R=0}}\\
    &+ \frac{1}{2}\trd{p_{rej}(\rho)\rho\up{2}_{TSABE,R=1} - p_{rej}(\sigma)\sigma\up{2}_{TSABE,R=1}}\\
                                & \ge \frac{1}{2}\trd{p_{ok}(\rho)\rho\up{2}_{TSABE,R=0} - p_{ok}(\sigma)\sigma\up{2}_{TSABE,R=0}}.
  \end{align*}
  Since $\mathcal{P}_\ell$ is a CPTP map, this further implies that
  \begin{equation}\label{eq:result:samp-diff}
    \trd{p_{ok}(\rho)\rho_{TSK_AEF,R=0} - p_{ok}(\sigma)\sigma_{TSK_AEF,R=0}}\le 2\sqrt{\epsilon^{cl}_\delta}.
  \end{equation}
  (Note that, above, we traced out $K_B$, but partial trace does not increase trace distance.)

  Now, let's look at the final ideal state, conditioned on not aborting, namely $\sigma_{TSK_AK_BER=0}$.  Since $\Hmin(A|E,R=0,T=t,S=s)_{\sigma^*} \ge \gamma$ for every $s\in\mathcal{S}$ and every subset $t$, we have, by Equation \ref{eq:main:PA}, that:
  \begin{equation}\label{eq:result:samp-ideal-pa}
    \trd{\sigma_{K_AEF,R=0,T=t,S=s} - \uni{\ell}\otimes \sigma_{EF,R=0,T=t,S=s}} \le 2^{-\frac{1}{2}(\gamma - \ell)}.
  \end{equation}

  Now, let $Q = (TSEF,R=0)$ for notational simplification.  Then, from the triangle inequality and basic properties of trace distance, we have:
  \begin{align}%
    &p_{ok}(\rho)\trd{\rho_{K_AQ} - \uni{\ell}\otimes\rho_{Q}}\label{generic:result:e1}\\
    &= ||\left(p_{ok}(\rho)\rho_{K_AQ} - p_{ok}(\sigma)\sigma_{K_AQ}\right)\notag\\
    &+ \left(p_{ok}(\sigma)\sigma_{K_AQ} - p_{ok}(\rho)\uni{\ell}\otimes\rho_{Q}\right)||\label{generic:result:e2}\\\notag\\
    &\le2\sqrt{\epsilon^{cl}_\delta} + ||\left(p_{ok}(\sigma)\sigma_{K_AQ} - p_{ok}(\sigma)\uni{\ell}\otimes\sigma_{Q}\right)\notag\\
    &+ \left(p_{ok}(\sigma)\uni{\ell}\otimes\sigma_{Q} - p_{ok}(\rho)\uni{\ell}\otimes\rho_{Q}\right)||\label{generic:result:e3}\\\notag\\
    & \le 2\sqrt{\epsilon^{cl}_\delta} + p_{ok}(\sigma)2^{-\frac{1}{2}(\gamma-\ell)}\notag\\
    &+ \trd{\left(p_{ok}(\sigma)\uni{\ell}\otimes\sigma_{Q} - p_{ok}(\rho)\uni{\ell}\otimes\rho_{Q}\right)}\label{generic:result:e4}\\\notag\\
    &\le 2^{-\frac{1}{2}(\gamma-\ell)} + 4\sqrt{\epsilon^{cl}_\delta}.\label{generic:result:e5}
  \end{align}
  Moving from Equations \ref{generic:result:e2} to \ref{generic:result:e3} follows from the triangle inequality and the bound in Equation \ref{eq:result:samp-diff}; moving from Equations \ref{generic:result:e3} to \ref{generic:result:e4} follows the triangle inequality and from Equation \ref{eq:result:samp-ideal-pa} (along with the fact that $p_{ok}(\cdot) \le 1$); and finally moving from Equation \ref{generic:result:e4} to \ref{generic:result:e5} follows, again, from Equation \ref{eq:result:samp-diff}, along with basic properties of trace distance.  This completes the proof.

%  Noting that $\Hmin(A|ETS,R=0)_{\sigma^*} \ge \min_{t; s\in\mathcal{S}} \Hmin(A_Z|E, T=t,S=s,R=0) \ge \gamma$ by the theorem hypothesis, we conclude:
%  \begin{equation}%\label{eq:result:}
%    p_{ok}(\rho)\trd{\rho_{K_ATSEF,R=0} - \uni{\ell}\otimes\rho_{TSEF,R=0}} \le 2^{-\frac{1}{2}(\Hmin(A|E)_{\sigma^*}-\ell)} + 4\sqrt{\epsilon^{cl}_\delta},
%  \end{equation}
%  as desired.
\end{proof}

\subsection{Min Entropy of HD-Bell States}%\label{sec:mainresult:}
We next show how to bound the min entropy of a superposition of Bell states if one knows the number of each Bell state appearing in a superposition.
\begin{lemma}\label{lemma:mainresult:min-ent-bell}
  Let $\ket{\psi}_{ABE} = \sum_{i\in J}\alpha_i \ket{\phi_i}_{AB}\ket{E_i}_E$, where $J = \{i \in \Bell^n \st \num_\alpha^\beta(i) = \lambda_\alpha^\beta\}$.  Assume a measurement is made of the $A$ and $B$ systems in the computational basis resulting in $\rho_{A_ZB_ZE}$.  Then:
  \begin{equation}%\label{eq:mainresult:}
    \Hmin(A_Z|E)_\rho \ge n\log_2d - \frac{1}{\log_d2}\sum_{\alpha=0}^{d-1}n_\alpha h_d\left(\frac{1}{n_\alpha}\sum_{\beta = 1}^{d-1}\lambda_\alpha^\beta\right),
  \end{equation}
  where: $n_\alpha = \sum_{\beta=0}^{d-1}\lambda_\alpha^\beta$.  Note that the summation in the entropy starts at one while all other summations start at zero and that $h_d(x)$ is the $d$-ary entropy function, defined in Section \ref{sec:main:prelim}.
\end{lemma}
\begin{proof}%\label{pf:mainresult:}
  Taking $\ket{\psi}_{ABE}$ and measuring in the computational basis yields the following mixed state:
  \begin{equation}%\label{eq:mainresult:}
    \rho_{A_ZB_ZE} = \sum_{a,b\in\al_d^n}\kb{a,b}P\left(\sum_{i\in J}\alpha_i\braket{a,b|\phi_i}\ket{E_i}\right).
  \end{equation}
  Recall that $\ket{\phi_{i_j}} = \frac{1}{\sqrt{d}}\sum_{x=0}^{d-1}\omega^{x\cdot i_j^\phase}\ket{x}_A\ket{x+i_j^\bit}_B$, where the arithmetic is done modulo $d$.  Thus:
  \begin{equation}%\label{eq:mainresult:}
    \braket{a_j,b_j|\phi_{i_j}} = \left\{\begin{array}{ll}
                                           \frac{1}{\sqrt{d}}\omega^{a_j\cdot i_j^\phase} & \text{ if } b_j = a_j + i_j^\bit \text{ (mod $d$)}\\
                                           0 & \text{ otherwise}
                                               \end{array}\right.
  \end{equation}
  From which it follows that $\braket{a,b|\phi_i} = \frac{1}{d^{n/2}}\omega^{a\cdot i^\phase}$ if $b = a \oplus_d i^\bit \iff i^\bit = b \ominus_d a$ and otherwise $\braket{a,b|\phi_i} = 0$.  Above, we have $a\cdot i^\phase = \sum_{j=0}^{n}a_j i_j^\phase$ and where $\oplus_d$  (respectively $\ominus_d$) is defined to be character-wise addition (respectively subtraction) modulo $d$.

  Thus, it follows that:
  \begin{equation*}%\label{eq:mainresult:}
    \rho_{A_ZB_ZE} = \sum_{a,b\in\al_d^n}\kb{a,b}\otimes P\left(\frac{1}{d^{n/2}}\sum_{\substack{i\in J\\\stt i^\bit = b\ominus_d a}}\alpha_i\omega^{a\cdot i^\phase}\ket{E_i}\right)
  \end{equation*}
  Tracing out $B_Z$ yields:
  \begin{equation*}\label{eq:mainresult:entropy-lemma-rho}
    \rho_{A_ZE} = \sum_{a\in\al_d^n}\kb{a}\otimes\sum_{b\in\al_d^n} P\left(\frac{1}{d^{n/2}}\sum_{\substack{i\in J\\\stt i^\bit = b\ominus_d a}}\alpha_i\omega^{a\cdot i^\phase}\ket{E_i}\right)
  \end{equation*}

  Now, consider, instead, the following state:
  \begin{equation}%\label{eq:mainresult:}
    \ket{\tau}_{ABE} = \sum_{i\in J}\alpha_i\ket{i^\bit}^Z_A\ket{i^\phase}^F_B\ket{E_i},
  \end{equation}
  where $F$ denotes the Fourier basis; namely for $j\in\al_d$, we define $\ket{j}^F = \frac{1}{\sqrt{d}}\sum_{x=0}^{d-1}\omega^{x\cdot j}\ket{x}$.  Clearly $\ket{\tau}$ is a valid quantum state as it can be constructed from $\ket{\psi}$ via the unitary operator mapping $\ket{\phi_{i_j}} \mapsto \ket{i_j^\bit}\ket{i_j^\phase}^F$ to each of the $n$ Bell pairs in the $AB$ register.  (Of course, Alice and Bob don't actually have to do this - we are simply arguing that $\ket{\tau}$ is a valid quantum state.)  We claim that, if the $A$ and $B$ registers of $\ket{\tau}$ are measured in the $Z$ basis, resulting in operator $\tau_{A_ZB_ZE}$, then it holds that $\Hmin(B_Z|E)_\tau = \Hmin(A_Z|E)_\rho$.  Thus, we will simply need to compute the min entropy of $\tau_{B_ZE}$ instead of $\rho_{A_ZE}$ which turns out to be easier.

  Measuring the $AB$ registers of $\ket{\tau}$ in the $Z$ basis yields:
  \begin{align*}%\label{eq:mainresult:}
    \tau_{A_ZB_ZE} &= \sum_{a,b\in\al_d^n}\kb{a,b}P\left(\sum_{i\in J}\alpha_i\braket{a|i^\bit}\braket{b|i^\phase}^F\ket{E_i}\right)\\
    &= \sum_{a,b}\kb{a,b}P\left(\sum_{\substack{i\in J\\\stt i^\bit = a}} \alpha_i\braket{b|i^\phase}^F\ket{E_i}\right).
  \end{align*}
  Now, since $\braket{b_j|i^\phase_j}^F = \frac{1}{\sqrt{d}}\omega^{b_j\cdot i_j^\phase}$, we have $\braket{b|i^\phase}^F = \frac{1}{d^{n/2}}\omega^{b\cdot i^\phase}$ and, so making this substitution into the above, and then tracing out $A_Z$ yields:
  \begin{equation}%\label{eq:mainresult:}
    \tau_{B_ZE} = \sum_{b}\kb{b}\sum_aP\left(\frac{1}{d^{n/2}}\sum_{\substack{i\in J\\\stt i^\bit = a}}\alpha_i\omega^{b\cdot i^\phase}\ket{E_i}\right).
  \end{equation}
  Now, re-ordering the above sum, we can write the state as:
    \begin{equation}%\label{eq:mainresult:}
    \tau_{B_ZE} = \sum_{b}\kb{b}\sum_xP\left(\frac{1}{d^{n/2}}\sum_{\substack{i\in J\\\stt i^\bit = x\ominus_d b}}\alpha_i\omega^{b\cdot i^\phase}\ket{E_i}\right).
  \end{equation}
  Comparing this, to Equation \ref{eq:mainresult:entropy-lemma-rho}, shows that $\Hmin(B_Z|E)_\tau$ will equal $\Hmin(A_Z|E)_\rho$.

  We must now compute $\Hmin(B_Z|E)_\tau$.  Returning to $\ket{\tau}$, before any measurements are made, we may write $\ket{\tau}$ in the following form: $\ket{\tau}_{ABE} = \sum_{i\in J}\alpha_i\ket{i^\bit}_A\ket{i^\phase}_B^F\ket{E_i}=$
  \begin{align*}%\label{eq:mainresult:}
    \sum_{i^\bit\in J^\bit}\beta_{i^\bit}\ket{i^\bit}_A\otimes\sum_{i^\phase\in J^\phase_{i^\bit}}\gamma_{i^\phase,i^\bit}\ket{i^\phase}_B^F\ket{E_i},
  \end{align*}
  where:
  \begin{align*}%
    J^\bit &= \left\{i^\bit\in\al_d^n \st \num_\alpha(i^\bit) = \sum_{\beta=0}^{d-1}\lambda_{\alpha}^\beta \text{ } \forall \alpha = 0, \cdots, d-1\right\}\\
    J^\phase_{i^\bit} &= \left\{i^\phase\in\al_d^n \st \num_\alpha^\beta{i^\phase \choose i^\bit} = \lambda_\alpha^\beta \text{ } \forall \alpha,\beta = 0, \cdots, d-1\right\}
  \end{align*}

  Tracing out the $A$ system above (which is equivalent to measuring it and then discarding the outcome) yields:
  \begin{equation}%\label{eq:mainresult:}
    \sigma_{BE} = \sum_{i^\bit\in J^\bit}|\beta_{i^\bit}|^2P\left(\sum_{i^\phase\in J^\phase_{i^\bit}}\gamma_{i^\phase,i^\bit}\ket{i^\phase}^F_B\ket{E_i}\right).
  \end{equation}
  
  By Lemmas \ref{lemma:main:entropy-super} and \ref{lemma:introduction:entropy-mixed}, along with Equation \ref{eq:main:min-ent-classical}, $\Hmin(B_Z|E)_\sigma \ge n\log_2d - \max_{i^\bit\in J^\bit}\log_2\left|J^\phase_{i^\bit}\right|$.  Thus, we must bound the size of $J^\phase_{i^\bit}$.  Fix a particular $i^\bit\in J^\bit$ and let $n_\alpha = \num_\alpha(i^\bit) = \sum_{\beta = 0}^{d-1}\lambda_{\alpha}^\beta$.  Then:
  \begin{align}%
    &|J^\phase_{i^\bit}| = \left|\left\{i^\phase\in\al_d^n \st \num_\alpha^\beta{i^\phase \choose i^\bit} = \lambda_\alpha^\beta\right\}\right|\notag\\
                        &=\left|\left\{i_{(0)}^\phase\in\al_d^{n_0}, \cdots, i_{(d-1)}^\phase\in\al_d^{n_{d-1}} \st \num_\beta(i_{(\alpha)}^\phase) = \lambda_\alpha^\beta\right\}\right|\label{eq:mainresult:lemma1:1}\\
                        &= \prod_{c=0}^{d-1}\left|\left\{i_{(c)}^\phase\in\al_d^{n_c} \st \num_\beta(i_{(c)}^\phase) = \lambda_c^\beta\right\}\right|\label{eq:mainresult:lemma1:2}\\
                        &\le \prod_c\left|\left\{i_{(c)}^\phase \in \al_d^{n_c} \st wt(i_{(c)}^\phase) \le \sum_{\beta=1}^{d-1}\lambda_c^\beta\right\}\right|\label{eq:mainresult:lemma1:3}\\
    &\le \prod_cd^{n_ch_d\left(\frac{1}{n_c}\sum_{\beta \ge 1}\lambda_c^\beta\right)}.\label{eq:mainresult:lemma1:4}
  \end{align}
  (note, above, $i_{(c)}$ is a bit string and does not mean the $c$'th character of $i$.)
  Above, Equation \ref{eq:mainresult:lemma1:1} follows by dividing $i^\phase$ into substrings based on the value of $i^\bit$ (which, for the bound above, is fixed).  Equation \ref{eq:mainresult:lemma1:2} follows by simply noting that the constraints on $i_{(j)}^\phase$ are independent of the constraints placed on $i_{(j')}^\phase$ for $j'\ne j$.  Equation \ref{eq:mainresult:lemma1:3} follows by removing specific constraints on the number of each character, and replacing it with a constraint, simply, on the number of non-zero characters (we comment that future work may improve our bound at this point).  Finally Equation \ref{eq:mainresult:lemma1:4} follows from the well-known bound on the volume of a Hamming ball.  The final bound above is independent of the choice of $i^\bit$.  From this, the lemma is proven.
\end{proof}

\subsection{Sampling Strategy}\label{sec:mainresult:sample-strat}

Our proof of security will use Theorem \ref{thm:introduction:sample} and, as a consequence, we must define and analyze an appropriate classical sampling strategy that correctly encapsulates the QKD protocol's testing stage.  We will be working with the HD-Bell basis as discussed in Section \ref{sec:main:prelim}.  As such our sampling will be with respect to words $q \in \Bell^N$.  The strategy will first choose a subset $t$ of size $m$ out of all possible $N$ systems (i.e., $t \subset \{1,\cdots, N\}$) and then choose a string $s \in \al_{d+1}^m$, where $Pr(s) = Pr(s_1)\cdots Pr(s_m)$ and $Pr(s_i = j) = 1/d+1$ for all $j = 0, 1, \cdots, d$.  This will simulate the testing subset choice as discussed earlier.  In the following, denote $m_j = |t_{s=j}|$, namely the size of the subset where $POVM$ $\Lambda^j$ will be used (see Section \ref{sec:protocol:}).

Now, normally, Alice and Bob will measure subset $t_{s=j}$ in POVM $\Lambda^j$ to observe the noise in the $j$'th basis.  This will be used to estimate the noise in the unobserved portion of the state.   In our case, we are trying to model this classically; for each $j=0,\cdots, d$ (possible basis) and $c=0,\cdots, d-1$ (possible outcome), define a set $P_c^j$ as follows:
\begin{equation}%\label{eq:mainresult:}
  P_c^j = \left\{(\alpha,\beta)\in\al_d\times\al_d \st \braket{\phi_\alpha^\beta|\Lambda^j_c|\phi_\alpha^\beta} \ge 0\right\}.
\end{equation}
Namely, $P_c^j$ is the set of all Bell states where the probability of observing POVM $\Lambda_c^j$, given that Bell state, is non-zero.
%We note that for for this protocol where $d$ is prime and the bases are the eigenvectors of the Wyl operators (CHECK), it holds that $P_c^j \cap P_{c'}^j = \emptyset$ for all $c\ne c'$ and, furthermore:
%\begin{equation}%\label{eq:mainresult:}
%  \Lambda_c^j\ket{\phi_i} = \left\{\begin{array}{ll}\alpha^i_{c,q}\ket{\phi_i} & \text{ if } q \in P_c^j\\
%    0 & \text{ otherwise}\end{array}\right.
%\end{equation}
%for some non-zero scalars $\alpha^i_{c,q}$.  The above follows since the Bell states are eigen-vectors of the chosen POVM operators (CITE).

From this, instead of defining the sampling strategy in terms of guess and target functions, we simply go straight to defining the set of good words (the guess and target functions can then be easily derived from this definition, however we feel defining the set of good words is more intuitive).  Note that our subset depends on $t$ and $s$, so the set of good words also depends on these two values (see our comment at the end of Section \ref{sec:intro:sampling}).  This set is denoted $\mathcal{G}^{t,s}$, and is the set of all $i \in \Bell^N$ such that:
\begin{equation}\label{eq:mainresult:good-words-actual}
\max_{j,c}\left|\frac{1}{m_j}\sum_{(\alpha,\beta)\in P_c^j} \num_\alpha^\beta(i_{t_{s=j}}) - \frac{1}{n}\sum_{(\alpha,\beta)\in P_c^j}\num_\alpha^\beta(i_{-t})\right| \le \delta
\end{equation}
where, again, $m_j = |t_{s=j}| = \num_j(s)$.  Note that, above, we only need to consider $c\ge 1$ as the $c=0$ case is then determined.  We need the error probability of this strategy as defined in Definition \ref{def:introduction:error-prob},
%\begin{equation}%\label{eq:mainresult:}
%\epsilon_\delta^{cl} = \max_{q\in\Bell^N}Pr\left(q \not\in \mathcal{G}^{t,s}_\delta\right).
%\end{equation}
where the probability is over all subsets $t$ and strings $s$.  To do so we will first consider the following, ``simpler'' strategy which only considers a single $j$ and $c$ (instead of all of them).   In particular, let $\mathcal{G}^{t,s,j,c}_\delta=$
  \begin{equation}\label{eq:mainresult:good-words-new-sample}
     \left\{q\in\Bell^N \st \left|\frac{1}{m_j}\sum_{(\alpha,\beta)\in P^j_c} \num_\alpha^\beta(q_{t_{s=j}}) - \frac{1}{n}\sum_{\alpha,\beta\in P^j_c}\num_\alpha^\beta(q_{-t})\right|\right\}
  \end{equation}
  Let $\epsilon_\delta^{j,c} = \max_{q\in\Bell^N} Pr_{t,s}\left(q \not \in \mathcal{G}^{t,s,j,c}\right)$.  From the union bound, it will follow that the failure probability of our desired sampling strategy (Equation \ref{eq:mainresult:good-words-actual}) will be upper-bounded by:
  \begin{equation}\label{eq:mainresult:sample-failure-pr}
    \epsilon_\delta^{cl} \le \max_{q\in\Bell^N}Pr\left(q \not\in \mathcal{G}^{t,s}_\delta\right) \le (d+1)(d-1)\max_{j,c}\epsilon_\delta^{j,c}.
  \end{equation}
  The above follows since there are $d-1$ characters to consider (not counting the zero character which is determined once we have the others) and there are $d+1$ bases.  We therefore must analyze the ``simple'' strategy which is considered below:

\begin{lemma}\label{lemma:mainresult:}
  Let $j \in \{0, 1, \cdots, d\}$, $c \in \{1,2,\cdots, d-1\}$, and $\delta > 0$.  Consider the simple sampling strategy discussed above, acting on words in $\Bell^N$ where the set of good words is defined in Equation \ref{eq:mainresult:good-words-new-sample}.  Then it holds that:
  \begin{align}
    \epsilon_\delta^{j,c} &\le 2\min_{\substack{\delta_1\in (0,\delta)\\\beta\in (0,1/(d+1))}}\left[ \exp\left(-\frac{(\delta-\delta_1)^2mN}{N+2}\right)\right.\notag\\
    &+\left. \exp\left(-\frac{\delta_1^2m^2\left(\frac{1}{d+1} - \beta\right)}{m+2}\right) + \exp\left(-2\beta^2m\right)\right],\label{eq:mainresult:sampling-main-error}
  \end{align}
  where the probability is over all subset choices $t \subset \{1, \cdots, N\}$ of size $m$ and also the choice of $s \in \al_{d+1}^m$ chosen uniformly at random.
\end{lemma}
\begin{proof}%\label{pf:mainresult:}
  First, consider instead the following sampling strategy over bit-strings: A subset $t$ is chosen along with $s$, same as above, but now we only care about the relative Hamming weight of the string indexed by $t_{s=j}$.  In this case, the set of good words, denoted $\widetilde{\mathcal{G}}^{t,s,j}_\delta$ is found to be:
  \begin{equation}%\label{eq:mainresult:}
    \widetilde{\mathcal{G}}^{t,s,j}_\delta = \left\{q \in \{0,1\}^N \st | w(q_{t_{s=j}}) - w(q_{-t})| \le \delta \right\}.
  \end{equation}
  We claim that:
  \begin{equation}\label{eq:mainresult:samp-lemma-1}
    \max_{q\in\Bell^N} Pr_{t,s}\left(q \not \in \mathcal{G}^{t,s,j,c}\right) \le \max_{q\in\{0,1\}^N}Pr\left(q \not\in \widetilde{\mathcal{G}}^{t,s,j}_\delta\right).
  \end{equation}
  Indeed, for any $q\in\Bell^N$, let $\widetilde{q}\in\{0,1\}^N$ such that $\widetilde{q}_i = 1$ only if $(q_i^\bit, q_i^\phase)\in P^j_c$.  Then, clearly:
  \begin{equation}%\label{eq:mainresult:}
    \frac{1}{m_j}\sum_{(\alpha,\beta)\in P^j_c} \num_\alpha^\beta(q_{t_{s=j}}) = \frac{1}{m_j}\num_1(\widetilde{q}_{t_{s=j}}) = w(\widetilde{q}_{t_{s=j}}).
  \end{equation}
  Similarly we have $\frac{1}{n}\sum_{\alpha,\beta}\num_{\alpha}^\beta(q_{-t}) = w(\widetilde{q}_{-t})$.  Thus, if $q\not \in \mathcal{G}^{t,s,j,c}_\delta$, then it follows $\widetilde{q}\not\in \widetilde{\mathcal{G}}^{t,s,j}_\delta$ from which Equation \ref{eq:mainresult:samp-lemma-1} follows.  We must therefore bound the failure probability of this sampling strategy over bit strings (i.e., the right-hand-side of Equation \ref{eq:mainresult:samp-lemma-1}).  For this, we will use arguments from Bouman and Fehr's paper \cite{bouman2010sampling} when they were analyzing various classical sampling strategies, namely their ``Example 3'' and ``Example 4'' strategies.  Our strategy is different from those analyzed in \cite{bouman2010sampling}, and will not follow immediately from their results, however, and so we must go through and prove it here formally.  The idea is that $t_{s=j}$ can be used to argue about the behavior of the word on $t$ which, itself, can be used to argue about the behavior of the word on $-t$.

  Fix $q\in\{0,1\}^N$ and fix some subset $t$ of size $m$.  First, we show that:
  %\begin{equation}%\label{eq:mainresult:}
$    Pr_s\left(\left| w(q_{t_{s=j}}) - w(q_t)\right| > \delta\right) \le \epsilon_1$
%  \end{equation}
  for some $\epsilon_1$ to be computed.  In particular, the weight within the subset $t_{s=j}$ is close to the weight within subset $t$ itself.  The probability, above, is over the choice of $s$, not $t$.

  Clearly, we have $\mathbb{E}(m_j) = \frac{m}{d+1}$ where, again, $m_j = |t_{s=j}| = \num_j(s)$.  By Hoeffding's Inequality, we have, for any $\beta \in (0, 1/(d+1))$:
  \begin{equation}%\label{eq:mainresult:}
    Pr_s\left(\left|\frac{m_j}{m} - \frac{1}{d+1}\right| \ge \beta\right) \le 2\exp(-2\beta^2m).
  \end{equation}
  Let $\mathcal{M} = \left\{m_j \st \left|\frac{m_j}{m} - \frac{1}{d+1}\right| \le \beta\right\}$.  From the above $Pr(m_j \not\in\mathcal{M}) \le 2\exp(-2\beta^2m)$.  Then:
  \begin{align*}%
    &Pr_s\left(\left|w(q_{t_{s=j}}) - w(q_t)\right| \ge \delta\right)\\
    &\le Pr_s(m_j\in\mathcal{M})Pr_s\left(\left|w(q_{t_{s=j}}) - w(q_t)\right| \st m_j \in \mathcal{M}\right)\\
    &+ Pr(m_j\not\in \mathcal{M})\notag\\
    &\le Pr_s\left(\left|w(q_{t_{s=j}}) - w(q_t)\right| \st m_j \in \mathcal{M}\right) + 2\exp(-2\beta^2m).
  \end{align*}
  Now, consider the above conditional probability.  This is equivalent to choosing a subset of $t$, of fixed size $m_j$, uniformly at random.  If $\beta < \frac{1}{2} - \frac{1}{d+1}$, then $m_j \le \frac{m}{2}$ (since $m_j\in\mathcal{M}$) and so Equation \ref{eq:introduction:basic-sample} can be applied to bound:
  \begin{align*}%
    &Pr_s\left(\left|w(q_{t_{s=j}}) - w(q_{t})\right| \st m_j\in\mathcal{M}\right)\\
    &\le 2\max_{m_j\in\mathcal{M}}\exp\left(-\frac{\delta^2m_jm}{m+2}\right)
    \le 2\exp\left(-\frac{\delta^2m^2\left(\frac{1}{d+1}-\beta\right)}{m+2}\right)
  \end{align*}
  and so:
  \begin{align}
    &Pr_s\left(\left|w(q_{t_{s=j}}) - w(q_t)\right| \ge \delta\right) \notag\\
    &\le 2\exp\left(-\frac{\delta^2m^2\left(\frac{1}{d+1}-\beta\right)}{m+2}\right) + 2\exp(-2\beta^2m).\label{eq:mainresult:sample-lemma-2}
  \end{align}

  Now, by the triangle inequality, we have for every $t, s$ and $q$:
  \begin{equation}%\label{eq:mainresult:}
\left|w(q_{t_{s=j}}) - w(q_{-t})\right| \le \left|w(q_{t_{s=j}}) - w(q_t)\right| + \left|w(q_t) - w(q_{-t})\right|.
  \end{equation}
  Let $\delta = \delta_1 + \delta_2$ (for $\delta_1,\delta_2 > 0$).  Then:
  \begin{align*}%
    &Pr_{t,s}\left(\left|w(q_{t_{s=j}}) - w(q_{-t})\right| \ge \delta_1+\delta_2\right)\\
    &\le  Pr\left(\left|w(q_{t_{s=j}}) - w(q_t)\right| + \left|w(q_t) - w(q_{-t})\right| \ge \delta_1+\delta_2\right)\le\\
    &Pr\left(\left|w(q_{t_{s=j}}) - w(q_t)\right| \ge \delta_1 \text{ or } \left|w(q_t) - w(q_{-t})\right| \ge \delta_2\right)\le\\
    & Pr\left(\left|w(q_{t_{s=j}}) - w(q_t)\right| \ge \delta_1\right) + Pr\left(\left|w(q_t) - w(q_{-t})\right| \ge \delta_2\right)
  \end{align*}

  Now, $Pr\left(\left|w(q_t) - w(q_{-t})\right| \ge \delta_2\right)$ does not depend on $s$ and so we may again use Equation \ref{eq:introduction:basic-sample} to find:
  \begin{equation}%\label{eq:mainresult:}
    Pr\left(\left|w(q_t) - w(q_{-t})\right| \ge \delta_2\right) \le 2\exp\left(-\frac{\delta_2^2mN}{N+2}\right).
  \end{equation}
  On the other hand, since $s$ and $t$ are independent, we have:
  %\begin{equation}%\label{eq:mainresult:}
    $Pr_{t,s}\left(\left|w(q_{t_{s=j}}) - w(q_t)\right| \ge \delta_1\right)$ is equal to  $\sum_tP_T(t)Pr_s\left(\left|w(q_{t_{s=j}}) - w(q_t)\right| \ge \delta_1\right),$
  %\end{equation}
  and, thus, Equation \ref{eq:mainresult:sample-lemma-2} applies.  Since the setting of $\beta$ and $\delta_1$ can be optimized by the user, and since $\delta_2 = \delta-\delta_1$, Equation \ref{eq:mainresult:sampling-main-error} follows, thus completing the proof.
\end{proof}

\subsection{Main Result}%\label{sec:mainresult:}

We are now, finally, in a position to state and prove our main theorem, namely a key-rate bound for the $(d+1)$-MUB protocol.  The proof is actually now very straight forward, given the above technical groundwork.

\begin{theorem}\label{thm:mainresult:main-thm}
  Let $\epsilon > 0$, $\delta > 0$, and let $\ket{\psi}_{ABE}$ be a quantum state, potentially produced by an adversary, where the $A$ and $B$ registers consist of $N$ qudits each (with each qudit of prime dimension $d \ge 2$).  Let $\mathcal{Q} = \{\hat{Q}_c^j\}$, for $c = 1, \cdots, d-1$ and $j=0, \cdots d$ be the maximal allowed error in each basis $j$ for each character $c$ (namely, the protocol aborts if the observed $w(q^j_c) > \hat{Q}^j_c$ (see Section \ref{sec:protocol:}).  Then, the protocol will produce an $(\epsilon+4\sqrt{\epsilon_\delta^{cl}})$-secure key (according to Definition \ref{def:secure} and where $\epsilon_\delta^{cl}$ is given by Equation \ref{eq:mainresult:sample-failure-pr}), of size $\ell$-bits, where $\ell = $
  \begin{equation}\label{eq:keyratefunction}
    n\log_2d - \frac{1}{\log_d2}\sum_{\alpha=0}^{d-1}n_\alpha h_d\left(\frac{1}{n_\alpha}\sum_{\beta=1}^{d-1}\lambda_\alpha^\beta\right) - \leakEC - 2\log\frac{1}{\epsilon},
  \end{equation}
  and where:
  \begin{equation}%\label{eq:mainresult:}
    n_\alpha = \sum_{\beta=0}^{d-1}\lambda_\alpha^\beta, \text{ and }     \lambda_\alpha^\beta = \frac{n}{d}\left( \hat{Q}_\alpha^0 + \sum_{s=0}^{d-1}\hat{Q}_{s\alpha-\beta \text{ mod } d}^{s+1} - 1\right)
  \end{equation}
  Above, $\leakEC$ is the error correction leakage.
%  and:
%  \begin{equation}%\label{eq:mainresult:}
%    \lambda_\alpha^\beta = \frac{n}{d}\left( \hat{Q}_\alpha^0 + \sum_{s=0}^{d-1}\hat{Q}_{s\alpha-\beta \text{ mod } d}^{s+1} - 1\right)
%  \end{equation}
\end{theorem}
\begin{proof}%\label{pf:mainresult:}
  Let $\rho_{TSABE} = \sum_{t,s}p(t,s)\kb{t,s}\otimes\kb{\psi}$ be the real state input to the QKD protocol, where the distribution over the $TS$ register is induced by the sampling strategy analyzed in Section \ref{sec:mainresult:sample-strat} (namely $p(t,s) = p(t)p(s)$ with $p(t)$ being the uniform distribution over subsets of size $m$ and with $s$ being chosen uniformly at random from all $s \in \al_{d+1}^m$; again, refer to the sampling strategy in Section \ref{sec:mainresult:sample-strat} for more details).  Let $\sigma_{TSABE} = \sum_{t,s}p(t,s)\kb{t,s}\kb{\phi\up{t,s}}_{ABE}$ be the ideal state constructed from Theorem \ref{thm:introduction:sample}, again using as a sampling strategy the one analyzed in Section \ref{sec:mainresult:sample-strat} and with respect to the high-dimensional Bell basis.   %The good states are constructed with respect to the Bell basis.  In particular, we have $\mathcal{G}^{t,s}$ from Equation \ref{eq:mainresult:good-words-actual}, but using the Bell basis for our quantum states.:
%  \begin{equation}%\label{eq:mainresult:}
%    \mathcal{G}^{t,s}_\delta = \left\{i \in \Bell^N \st \left| \frac{1}{m_j}\sum_{(\alpha,\beta)\in P_C^j}\num_\alpha^\beta(i_{t_{s=j}}) - \frac{1}{n}\sum_{(\alpha,\beta)\in P_C^j}\num_\alpha^\beta(i_{-t})\right| \le \delta\right\}
%  \end{equation}
%  where, recall, $m_j = |t_{s=j}| = \num_j(s)$.

  We compute the min entropy of the ideal state first; then we may use Lemma \ref{lemma:result:security-chain} to promote this to an analysis of the given real input state $\ket{\psi}_{ABE}$.  We first run the sampling strategy on the ideal state $\sigma_{TSABE}$; the choice of subset $t$ and $s$ causes the ideal state to collapse to $\ket{\phi\up{t,s}}_{ABE}$.  Measuring $t_{s=j}$, for $j=0, \cdots, d$ using POVM $\Lambda^j$ (see Section \ref{sec:protocol:}) results in outcome $q = (q^0,q^1,\cdots, q^d)$, with each $q^j \in \al_d^{m_j}$.  This results in the post measured state $\sigma\up{t,s,q}_{ABE}$ (conditioning on a particular $t,s$ and $q$), where the $A$ and $B$ registers consist of $n = N-m$ qudits each.  Since the initial state was an ideal state, living in the spanning set induced by Equation \ref{eq:mainresult:good-words-actual}, we may write this state in the following form (after tracing out the measured system and permuting subspaces):
  \begin{equation}%\label{eq:mainresult:}
    \sigma\up{t,s,q}_{ABE} = \sum_{i\in J_q}p_iP\left(\sum_{k\in\mathcal{I}_{q,i}}\alpha_{k,t,s,q,i}\ket{\phi_k}\ket{E_{k,t,s,q,i}}\right),
  \end{equation}
  where:
  \begin{equation}%\label{eq:mainresult:}
    J_q = \left\{i\in\Bell^m \st \sum_{(\alpha,\beta)\in P_c^j}\num_\alpha^\beta(i_{s=j}) = \num_c(q^j)\right\}
  \end{equation}
  and $\mathcal{I}_{q,i} $
  %\begin{align*}%\label{eq:mainresult:}
    %&\{k \in \Bell^n \st |\frac{1}{n}\sum_{(\alpha,\beta)\in P_c^j}\num_\alpha^\beta(k) - \frac{1}{m_j}\sum_{(\alpha,\beta)\in P_c^j}\num_\alpha^\beta(i_{{s=j}})| \le\delta\}\\
    $= \{k \in \Bell^n \st \frac{1}{n}\sum_{(\alpha,\beta)\in P_c^j}\num_\alpha^\beta(k) = Q_c^j\},$
  %\end{align*}
  where $Q_c^j = \frac{1}{m_j}\num_c(q^j) + \delta_c^j$ for $\delta_c^j \in [-\delta,\delta]$.

  Now, Alice and Bob will abort if $q^j_c > \hat{Q}_c^j$, where the latter are provided by the users and specified in $\mathcal{Q}$.  Thus, conditioning on not aborting, the state can be written:
  \begin{equation*}%\label{eq:mainresult:}
    \sigma\up{2}_{TSQ} = \frac{1}{p_{ok}(\sigma)}\sum_{t,s}p(t,s)\kb{t,s}_{TS}\sum_{q\in\mathcal{Q}}p(q|t,s)\kb{q}_Q\otimes\sigma_{ABE}\up{t,s,q}
  \end{equation*}
  where $p_{ok}(\sigma)$ is the probability of not aborting the protocol, given the ideal state (see Section \ref{sec:mainresult:ideal-promote}).

  A nice property of the MUBs chosen is that, given $\frac{1}{n}\sum_{(\alpha,\beta)}\num_\alpha^\beta(k) = Q_c^j$, for every $c, j$, then one can solve for each $\frac{1}{n}\num_\alpha^\beta(k)$ (for each $\alpha,\beta \in \al_d$).  In particular, it can be shown that:
  %\begin{equation}%\label{eq:mainresult:}
   $ \frac{1}{n}\num_\alpha^\beta(k) = \frac{1}{d}\left( Q_\alpha^0 + \sum_{s=0}^{d-1}Q_{s\alpha-\beta \text{ mod } d}^{s+1} - 1\right).$
  %\end{equation}
  (See \cite{sheridan2010security,wyderka2025high} for details on how the above equation is derived from the observed statistics.)

  Let $\lambda_\alpha^\beta = \frac{n}{d}\left( Q_\alpha^0 + \sum_{s=0}^{d-1}Q_{s\alpha-\beta \text{ mod } d}^{s+1} - 1\right)$.  Then, after measuring the $A$ register in the $Z$ basis (resulting in register $A_Z$) it holds that, from Lemma \ref{lemma:mainresult:min-ent-bell} that $ \Hmin(A_Z|E)_{\sigma\up{t,s,q}} \ge \gamma$ with:
  %\begin{equation}\label{eq:mainresult:thmmain-min-ent}
   $\gamma = n\log_2d - \frac{1}{\log_d2}\sum_{\alpha=0}^{d-1}n_\alpha h_d\left(\frac{1}{n_\alpha}\sum_{\beta = 1}^{d-1}\lambda_\alpha^\beta\right),$
 % \end{equation}
  where:
  %\begin{equation}%\label{eq:mainresult:}
    $n_\alpha = \sum_{\beta=0}^{d-1}\lambda_\alpha^\beta.$
%  \end{equation}

  Since the above expression for $\gamma$ is clearly minimized when $Q_c^j = \hat{Q}_c^j + \delta$ for every $j$ and $c \ge 1$ (i.e., we assume the worst case, that the noise is as high as possible before Alice and Bob abort), Lemma \ref{lemma:result:security-chain} can then be used to promote the analysis to the real state $\rho_{TSABE}$ (we set the right-hand side of Equation \ref{eq:result:lemma-chain} to $\epsilon + 4\sqrt{\epsilon_\delta^{cl}}$ and solve for $\ell$ giving us Equation \ref{eq:keyratefunction} as desired).  Of course, we must deduct $\leakEC$ from the final entropy bound as Eve has this information; this completes the proof.
\end{proof}

%%% Local Variables:
%%% mode: latex
%%% TeX-master: "main"
%%% End:

\section{Evaluation}%\label{sec:evaluation:}
% \subsection{$\delta_{min}$}%\label{subsec:deltamin}
We now evaluate Equation \ref{eq:keyratefunction} and compare with prior work in \cite{wyderka2025high} (which also considered asymmetric channels, but used an alternative proof methodology).  To evaluate our bound, we take as input the desired security level $\epsilon$ (which we take to be $10^{-12}$).  We then find a suitable $\delta$ such that the overall error probability (taking into account our bound on $\epsilon_\delta^{cl}$ from Equation \ref{eq:mainresult:sample-failure-pr} and our Theorem \ref{thm:mainresult:main-thm}), is no greater than the given overall security level $\epsilon$.  Note that $\epsilon^{cl}_\delta$ is a minimization over two parameters.  To make this tractable, we use values for $\delta_1$ and $\beta$ that seem to produce good results for the cases we evaluate.  We comment that alternative values may be found which can only improve our results presented here.

More formally, given the overall desired security level of $\epsilon_{sec} = \epsilon + 4\sqrt{\epsilon^{cl}_\delta}$ (such that the final secret key is $\epsilon_{sec}$-secure according to Definition \ref{def:secure}), with some user specified $\epsilon_{sec} > \epsilon > 0$, we must find the smallest $\delta$ such that $\epsilon + 4(d+1)(d-1)\sqrt{\epsilon_\delta^{j,c}} \le \epsilon_{sec}$, where $\epsilon_\delta^{j,c}$ is from Equation \ref{eq:mainresult:sampling-main-error}.
%We want to maximize the keyrate. $\delta$ and the keyrate are inversely related. So to maximize the keyrate, we must minimize $\delta$ to give us $\delta_{min}$. However, $\delta$ and $\epsilon_\delta$ are exponentially inversely proportional, so minimizing $\delta$ also maximizes $\epsilon_\delta$. So to ensure the key is secure enough, we find the minimum delta s.t. 
%\begin{equation}\label{eq:eps-delt-constraint}
%  \epsilon_\delta^{j,c} \leq \frac{(\epsilon_{sec} - \epsilon)^2}{16(d+1)(d-1)} 
%\end{equation}
%for $\epsilon  < \epsilon_{sec}$ and $\epsilon_\delta^{j,c}$ from Equation (REF).
In our evaluations, we use $\epsilon = 10^{-14}$ and $\epsilon_{sec} = 10^{-12}$.
To simplify the computation of $\epsilon_\delta^{j,c}$, we use $\beta = 1/d^2$ and $\delta_1 = \delta\cdot c$, for some constant $c\in(0,1)$ to be determined.  Note that $\epsilon_\delta^{j,c}$ is an expression of three terms, only two of which depend on $\delta$.  Based on which term dominates, then we can solve for $\delta$ by taking $\epsilon^{j,c}_\delta$ to be simply three times the dominating term, then solving for $\delta$ directly.  Note, future work in performing an optimization can only improve the results we present here, as smaller $\delta$'s may be found (thus higher key-rates as $\delta$ determines the confidence interval on our sampling set).

We note the third term in $\epsilon_\delta^{j,c}$ does not depend on $\delta$ and can always me made smaller than the other two which will increase as $\delta$ decreases. 
Let the first two terms of $\epsilon_\delta^{j,c}$ be represented by $X$ and $Y$, respectively. Then, we find $X < Y$ whenever
%\begin{equation}\label{eq:eps-delt-constraint2}
$    c < c_\gamma:= \left(\sqrt{\frac{m}{m+2}\frac{N+2}{N}\left(\frac{1}{d+1}-\frac{1}{d^2}\right)}+1\right)^{-1}$
%\end{equation}
%Although $N$ and $m$ don't affect $c_d$ greatly, they do affect this inequality. 
%So in order for us to use $c_d$, we must turn $\delta_{min}$ into the following piecewise function:
%Based on which of the two terms dominates, we have one of two possible $\delta$ values which minimize $\epsilon_\delta^{cl}$.  They are:
%Let $e^\alpha$ be the dominating term.
%We solve for $\alpha$ by substituting $\epsilon_\delta^{cl}$ with $6e^\chi$ in Equation \ref{eq:eps-delt-constraint}:
%\[\chi \leq \chi_d =\ln\left(\frac{(\epsilon_{sec} - \epsilon)^2}{96(d+1)(d-1)}\right)\]
%We expand $\chi$ using the first term to get Equation \ref{eq:1st-dominates}, or the second (middle) term to get Equation \ref{eq:2nd-dominates}.
Thus, we set $\delta = \delta_{min}$:
\begin{equation}
    \delta_{min} \geq
    \begin{cases} 
        \sqrt{-\frac{\chi_d}{(1-c)^2}\frac{N+2}{Nm}} & c_d \geq c_\gamma \\\\
        \sqrt{-\frac{\chi_d}{c^2}\frac{m+2}{m^2}\left(\frac{1}{d+1}-\frac{1}{d^2}\right)^{-1}} & c_d < c_\gamma 
     \end{cases}
   \end{equation}
where $\chi_d =\ln\left((\epsilon_{sec} - \epsilon)^2/(96(d+1)(d-1))\right)$.    Then it will hold that $\epsilon + 4\sqrt{\epsilon_{\delta}^{cl}} = \epsilon_{sec}$.
This will allow us to evaluate our key-rate bound rapidly.  In particular, we evaluate $\texttt{rate} = r = \ell/N$ using our Theorem \ref{thm:mainresult:main-thm} and compare to Prior Work in \cite{wyderka2025high}.  For both cases, we optimize over the sample size $m$ to maximize the key-rates.

\textbf{Symmetric Noise Results: }
We first evaluate our result on symmetric channels, assuming a depolarization channel.  Under such a channel, we would expect $\hat{Q}^j_c = Q/(d-1)$ for all $c \ne 0$ and for all bases $j$, where $Q$ is the depolarization probability.  We compare with results in \cite{wyderka2025high}, which, to our knowledge, represents the best known key-rate bound which is capable of handling arbitrary (symmetric and asymmetric) channels.

As shown in Figure \ref{fig:RvsN}, we see that our results outperform prior work \cite{wyderka2025high} whenever the number of signals, $N$, is low or the dimension is high.  For other cases, prior work outperforms.  Our result perform better at lower signals or higher dimensions due to the fact that we do not rely on AEP or post selection techniques as in prior work; for higher signals, our sampling proof strategy is not as optimal as the one in \cite{wyderka2025high} and so we take a performance hit (as we must raise the error probability to the $1/2$ power, causing $\delta$ to increase, which is an artifact of our proof technique).

We show in Figure \ref{fig:SymmetricPlots} that our result is more tolerant of noise by comparison at higher $d$ and lower $N$, and that difference decreases as $N$ increases, due to the two functions being asymptotically equivalent.
%%We see in the top-right plot that ours reaches its comparative best at low $Q$, whereas the bottom-right graph shows prior work performs its comparative best around the midpoint of the upper and lower noise bounds.

\textbf{Asymmetric Noise Results: }
We now test asymmetric noise, namely when the error rate is different across different bases.  We assume, for any $j$, that $\hat{Q}^j_c = \hat{Q}^j_{c'}$ for $c\ne c'$ and both non-zero; however, now we have $\hat{Q}^j_c \ne \hat{Q}^{j'}_c$ for $j \ne j'$. 
Let $\hat{Q}^\beta$ be the independent variable (i.e., the basis with varying noise), while $\hat{Q}^x$ is held constant for all $x\neq\beta$.
We've found that our function behaves differently depending on whether $\beta=0$, $\beta=1$, or $\beta>1$, as shown in Figure \ref{fig:RvsAsymmQ_variedBasis}.
%As shown, ours keyrate function outperforms prior work when $\hat{Q}^\beta < \hat{Q}^X$.
%Then past the point of intersection (i.e., when $\hat{Q}^\beta > \hat{Q}^X$, ours out performs prior work, when $\beta=1$, followed by $\beta=0$ and $\beta>1$.
%$\beta=1$ is different from every other case $\beta \neq 1$ due to $\sum_{\beta=1}^{d-1} \lambda_\alpha^\beta$ in Equation \ref{eq:keyratefunction}.
%We can rewrite this equation as $n_\alpha - \lambda_\alpha^0$. Then for $\alpha \neq 0$,
%\[\lambda_\alpha^0 = \frac{Q^0}{d-1} - Q^1 + \sum_{s=1}^{d-1}\frac{Q^{s+1}}{d-1} + \zeta_\alpha\]
%where $\zeta_\alpha = \delta_\alpha^0 + \sum_{s=0}^{d-1}\delta_{s\alpha}^{s+1}$.
%This is the same for every $\alpha$ in Equation \ref{eq:keyratefunction}, so the functions are weighed in favor of $Q^1$ compared to $Q^\beta$ for $\beta > 2$.

%Like with symmetric noise, we find that our plot and the prior work's plot converge asymptotically as $N$ increases for every value of $Q^0$.
%As in figure \ref{fig:RvsAsymmQ_variedN}, we see the two plots overlap more at $N=10^9$ than they do at $N=10^7$.

As in the symmetric case, we see that our result tends to outperform prior work at lower $N$ and higher $d$. However, prior work begins to outperform ours in other scenarios.  However, since both are lower bounds, users may simply take the max of our result and prior work to find a good bound on the key-rate in all scenarios.
%%So as the two begin to approach the asymptotic value, we find that ours works better at lower $Q$, theirs works better at higher $Q$.
%$N=10^14$ shows that those differences get smaller over time, suggesting that they're asymptotically equivalent in the asymmetric case as well.

Finally, and interestingly, we also observe similar behavior in the asymmetric case that was discovered in \cite{murta2020key} for the two dimensional case.  In particular, we note that increased noise in some bases do not always lead to a smaller key-rate.  Our evaluations, here, generalize the findings in \cite{murta2020key} to higher dimension.

\begin{figure}
    \centering
    \includegraphics[width=0.95\linewidth]{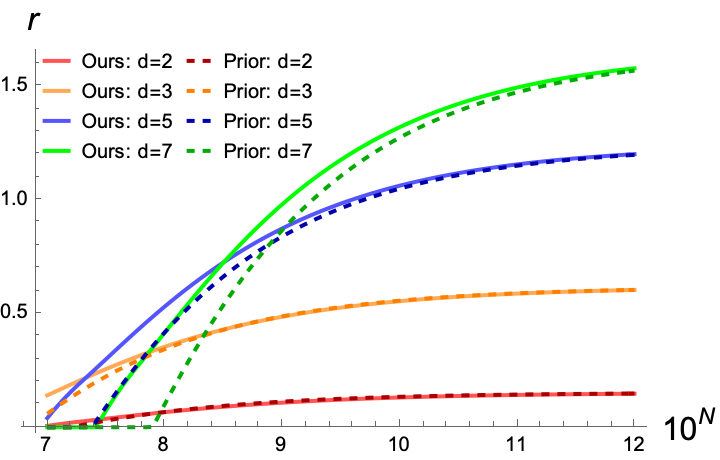}
    \caption{Keyrate vs. number of signals $N$ for various $d$, setting $Q=0.1$ in the symmetric noise case, comparing our new keyrate function with that in prior work \cite{wyderka2025high}. 
    We note that our function can produce a positive keyrate at lower signals than prior work can.}
    \label{fig:RvsN}
\end{figure}

\begin{figure}
    \centering
    \includegraphics[width=0.5\linewidth]{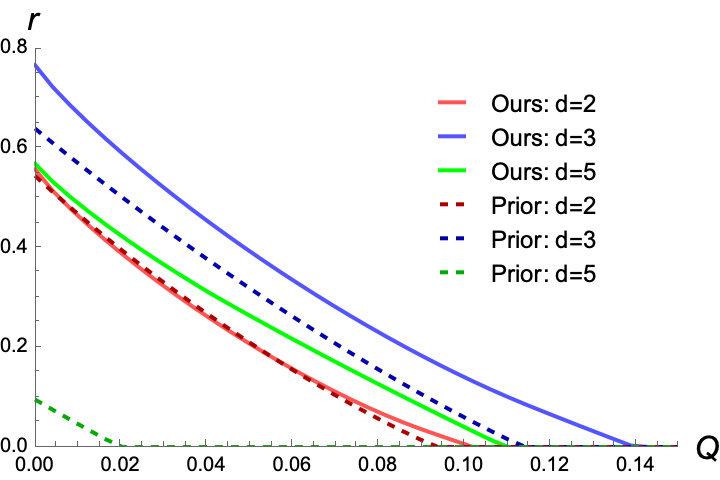}%
    \includegraphics[width=0.5\linewidth]{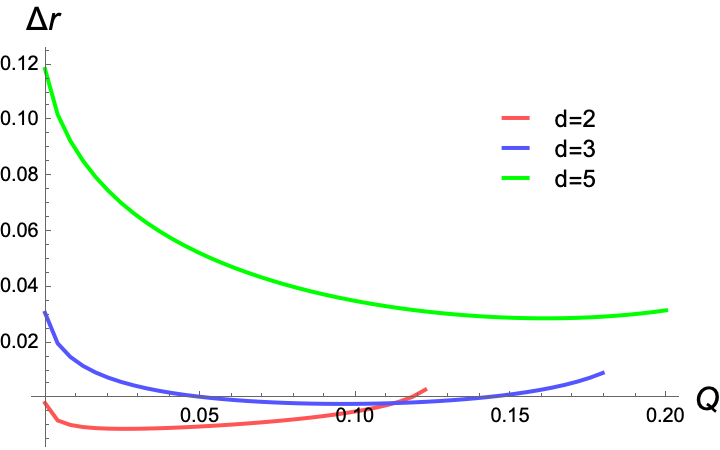}
    \caption{Keyrate vs. Noise (symmetric case $Q$) at various dimensions. 
    Left: Key-rates when $N=10^7$. Right: Difference between our key-rate and prior work when $N=10^9$ (positive values imply ours outperforms while negative values implies prior work outperforms).
    On the left plot, we see our work is more tolerant of noise at this setting of $N$.
    On the right plot, we see there are settings where prior work outperforms our work.  Since both results are lower bounds, however, users may simply take the maximum of our key-rate and prior work.
    }
    \label{fig:SymmetricPlots}
\end{figure}

\begin{figure}
    \centering
    \includegraphics[width=0.5\linewidth]{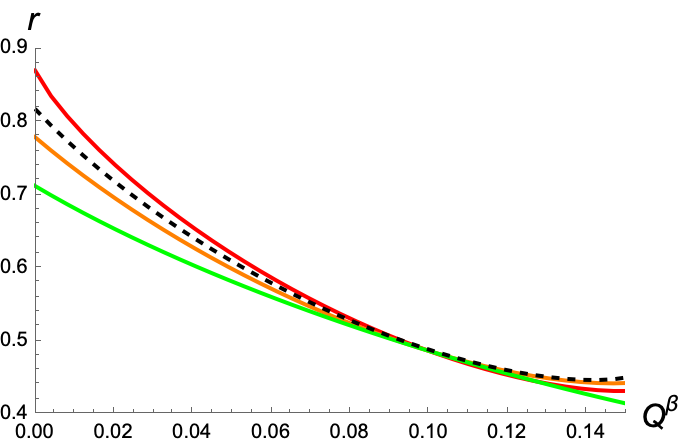}%
    \includegraphics[width=0.5\linewidth]{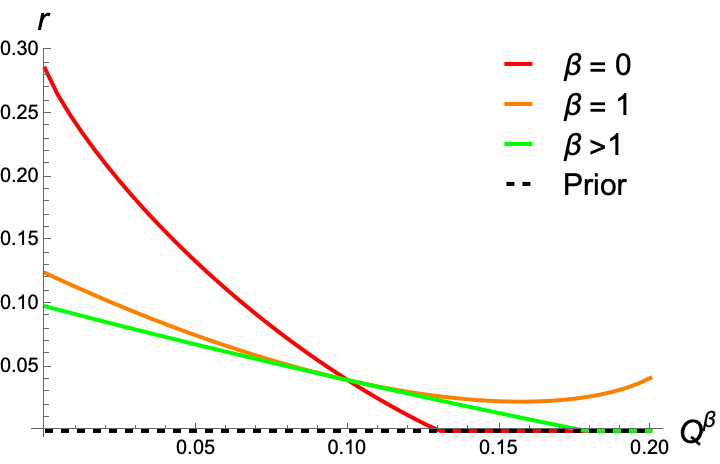}
    \caption{Keyrate vs. Noise in the asymmetric case. 
      The left graph uses $d=3$, $N=10^9$. The right graph uses $d=5$, $N=10^7$.  We vary the noise in a single basis (the $\beta$ basis) while keeping all other basis noise fixed at $10\%$.
    Note that in the right plot, the prior work key-rate evaluates to zero, implying users should abort.  Also note the slight increase in key-rate as the noise increases for the $\beta=1$ case, which follows a trend discovered for the two-dimensional case, in \cite{murta2020key}.
    }
    \label{fig:RvsAsymmQ_variedBasis}
\end{figure}

%\begin{figure}
%    \centering
%    \includegraphics[width=0.95\linewidth]{Graphs/RvsAsymQ_VariedN_d=3_2.png}
%    \caption{Keyrate vs. Asymmetric noise in the keyrate basis at different $N$. $Q^\beta=0.1$ for $\beta>0$, $d=3$. Ours performs better at lower $N$. 
%    As $N$ increases, they asymptotically approach the same value at every $Q^0$.
%    There is a range of $N$ for which theirs performs better at $Q^0 > 0.1$.
%    }
%    \label{fig:RvsAsymmQ_variedN}
%\end{figure}

%%% Local Variables:
%%% mode: LaTeX
%%% TeX-master: "main"
%%% End:

\section{Closing Remarks}
In this paper, we derived a novel proof of security for the HD-QKD protocol utilizing $d+1$ MUBs originally proposed in \cite{cerf2002security}.  Our proof of security can handle any arbitrary, general, attack and channel including  asymmetric noise.  Our proof bounds the min entropy directly, without having to rely on AEP or post selection methods as discussed in the Introduction.  We compared to prior work in \cite{wyderka2025high} (which to our knowledge is the best known key-rate bound for this protocol that is able to handle arbitrary, asymmetric, channels and attacks). Our evaluations showed that our result can outperform when the number of signals is low or the dimension is high; otherwise prior work outperforms our work.  However, since both are lower-bounds, users may take the maximum of both to derive a final key rate, and thus these two papers complement one another.  Note that both results converge to one another asymptotically.

Many interesting future problems remain open.  Deriving a general entropic uncertainty relation for multiple bases would be very useful and our proof method may help with this.  Analyzing multi-qubit sources and lossy channels would also be an important next step.

\textbf{Acknowledgments:} The authors would like to acknowledge support from the NSF under grant number 2143644.

\balance
%\bibliographystyle{unsrt}
%\bibliography{local}

\begin{thebibliography}{10}

\bibitem{QKD-survey}
Valerio Scarani, Helle Bechmann-Pasquinucci, Nicolas~J. Cerf, Miloslav
  Du\ifmmode~\check{s}\else \v{s}\fi{}ek, Norbert L\"utkenhaus, and Momtchil
  Peev.
\newblock The security of practical quantum key distribution.
\newblock {\em Rev. Mod. Phys.}, 81:1301--1350, Sep 2009.

\bibitem{pirandola2020advances}
Stefano Pirandola, Ulrik~L Andersen, Leonardo Banchi, Mario Berta, Darius
  Bunandar, Roger Colbeck, Dirk Englund, Tobias Gehring, Cosmo Lupo, Carlo
  Ottaviani, et~al.
\newblock Advances in quantum cryptography.
\newblock {\em Advances in optics and photonics}, 12(4):1012--1236, 2020.

\bibitem{amer2021introduction}
Omar Amer, Vaibhav Garg, and Walter~O Krawec.
\newblock An introduction to practical quantum key distribution.
\newblock {\em IEEE Aerospace and Electronic Systems Magazine}, 36(3):30--55,
  2021.

\bibitem{broadbent2020quantum}
Anne Broadbent and Rabib Islam.
\newblock Quantum encryption with certified deletion.
\newblock In {\em Theory of Cryptography: 18th International Conference, TCC
  2020, Durham, NC, USA, November 16--19, 2020, Proceedings, Part III 18},
  pages 92--122. Springer, 2020.

\bibitem{broadbent2016quantum}
Anne Broadbent and Christian Schaffner.
\newblock Quantum cryptography beyond quantum key distribution.
\newblock {\em Designs, Codes and Cryptography}, 78:351--382, 2016.

\bibitem{bozzio2024quantum}
Mathieu Bozzio, Claude Cr{\'e}peau, Petros Wallden, and Philip Walther.
\newblock Quantum cryptography beyond key distribution: theory and experiment.
\newblock {\em arXiv preprint arXiv:2411.08877}, 2024.

\bibitem{cerf2002security}
Nicolas~J Cerf, Mohamed Bourennane, Anders Karlsson, and Nicolas Gisin.
\newblock Security of quantum key distribution using d-level systems.
\newblock {\em Physical review letters}, 88(12):127902, 2002.

\bibitem{bechmann2000quantum}
Helle Bechmann-Pasquinucci and Wolfgang Tittel.
\newblock Quantum cryptography using larger alphabets.
\newblock {\em Physical Review A}, 61(6):062308, 2000.

\bibitem{sheridan2010security}
Lana Sheridan and Valerio Scarani.
\newblock Security proof for quantum key distribution using qudit systems.
\newblock {\em Physical Review A—Atomic, Molecular, and Optical Physics},
  82(3):030301, 2010.

\bibitem{vlachou2018quantum}
Chrysoula Vlachou, Walter Krawec, Paulo Mateus, Nikola Paunkovi{\'c}, and
  Andr{\'e} Souto.
\newblock Quantum key distribution with quantum walks.
\newblock {\em Quantum Information Processing}, 17(11):288, 2018.

\bibitem{iqbal2021analysis}
Hasan Iqbal and Walter~O Krawec.
\newblock Analysis of a high-dimensional extended b92 protocol.
\newblock {\em Quantum Information Processing}, 20:1--22, 2021.

\bibitem{cozzolino2019high}
Daniele Cozzolino, Beatrice Da~Lio, Davide Bacco, and Leif~Katsuo Oxenl{\o}we.
\newblock High-dimensional quantum communication: benefits, progress, and
  future challenges.
\newblock {\em Advanced Quantum Technologies}, 2(12):1900038, 2019.

\bibitem{renner2005information}
Renato Renner, Nicolas Gisin, and Barbara Kraus.
\newblock Information-theoretic security proof for quantum-key-distribution
  protocols.
\newblock {\em Physical Review A—Atomic, Molecular, and Optical Physics},
  72(1):012332, 2005.

\bibitem{wyderka2025high}
Nikolai Wyderka, Giovanni Chesi, Hermann Kampermann, Chiara Macchiavello, and
  Dagmar Bru{\ss}.
\newblock High-dimensional quantum key distribution rates for multiple
  measurement bases.
\newblock {\em arXiv preprint arXiv:2501.05890}, 2025.

\bibitem{bouman2010sampling}
Niek~J Bouman and Serge Fehr.
\newblock Sampling in a quantum population, and applications.
\newblock In {\em Annual Cryptology Conference}, pages 724--741. Springer,
  2010.

\bibitem{krawec2019quantum}
Walter~O Krawec.
\newblock Quantum sampling and entropic uncertainty.
\newblock {\em Quantum Information Processing}, 18(12):368, 2019.

\bibitem{yao2022quantum}
Keegan Yao, Walter~O Krawec, and Jiadong Zhu.
\newblock Quantum sampling for finite key rates in high dimensional quantum
  cryptography.
\newblock {\em IEEE Transactions on Information Theory}, 68(5):3144--3163,
  2022.

\bibitem{krawec2023entropic}
Walter~O Krawec.
\newblock Entropic uncertainty for biased measurements.
\newblock In {\em 2023 IEEE International Conference on Quantum Computing and
  Engineering (QCE)}, volume~1, pages 1220--1230. IEEE, 2023.

\bibitem{tomamichel2009fully}
Marco Tomamichel, Roger Colbeck, and Renato Renner.
\newblock A fully quantum asymptotic equipartition property.
\newblock {\em IEEE Transactions on information theory}, 55(12):5840--5847,
  2009.

\bibitem{christandl2009postselection}
Matthias Christandl, Robert K{\"o}nig, and Renato Renner.
\newblock Postselection technique for quantum channels with applications to
  quantum cryptography.
\newblock {\em Physical review letters}, 102(2):020504, 2009.

\bibitem{wang2021tight}
Rong Wang, Zhen-Qiang Yin, Hang Liu, Shuang Wang, Wei Chen, Guang-Can Guo, and
  Zheng-Fu Han.
\newblock Tight finite-key analysis for generalized high-dimensional quantum
  key distribution.
\newblock {\em Physical Review Research}, 3(2):023019, 2021.

\bibitem{murta2020key}
Gl{\'a}ucia Murta, Filip Rozpedek, J{\'e}r{\'e}my Ribeiro, David Elkouss, and
  Stephanie Wehner.
\newblock Key rates for quantum key distribution protocols with asymmetric
  noise.
\newblock {\em Physical Review A}, 101(6):062321, 2020.

\bibitem{renner2008security}
Renato Renner.
\newblock Security of quantum key distribution.
\newblock {\em International Journal of Quantum Information}, 6(01):1--127,
  2008.

\bibitem{konig2009operational}
Robert Konig, Renato Renner, and Christian Schaffner.
\newblock The operational meaning of min-and max-entropy.
\newblock {\em IEEE Transactions on Information theory}, 55(9):4337--4347,
  2009.

\bibitem{tomamichel2012tight}
Marco Tomamichel, Charles Ci~Wen Lim, Nicolas Gisin, and Renato Renner.
\newblock Tight finite-key analysis for quantum cryptography.
\newblock {\em Nature communications}, 3(1):634, 2012.

\bibitem{bialynicki2011entropic}
Iwo Bialynicki-Birula and {\L}ukasz Rudnicki.
\newblock Entropic uncertainty relations in quantum physics.
\newblock {\em Statistical Complexity: Applications in Electronic Structure},
  pages 1--34, 2011.

\bibitem{coles2017entropic}
Patrick~J Coles, Mario Berta, Marco Tomamichel, and Stephanie Wehner.
\newblock Entropic uncertainty relations and their applications.
\newblock {\em Reviews of Modern Physics}, 89(1):015002, 2017.

\bibitem{wehner2010entropic}
Stephanie Wehner and Andreas Winter.
\newblock Entropic uncertainty relations—a survey.
\newblock {\em New Journal of Physics}, 12(2):025009, 2010.

\bibitem{wootters1989optimal}
William~K Wootters and Brian~D Fields.
\newblock Optimal state-determination by mutually unbiased measurements.
\newblock {\em Annals of Physics}, 191(2):363--381, 1989.

\bibitem{durt2010mutually}
Thomas Durt, Berthold-Georg Englert, Ingemar Bengtsson, and Karol
  {\.Z}yczkowski.
\newblock On mutually unbiased bases.
\newblock {\em International journal of quantum information}, 8(04):535--640,
  2010.

\bibitem{krawec2022security}
Walter~O Krawec.
\newblock Security of a high dimensional two-way quantum key distribution
  protocol.
\newblock {\em Advanced Quantum Technologies}, 5(10):2200024, 2022.

\end{thebibliography}

\end{document}